\documentclass[twocolumn,nobibnotes,aps,pra,reprint,superscriptaddress]{revtex4-1}

\usepackage{microtype}
\usepackage{mathtools}
\usepackage{array}
\usepackage{braket}
\usepackage{bbm}
\usepackage{amsthm}
\usepackage{slashed}

\newcommand{\ketbra}[2]{\ket{#1}\!\bra{#2}}

\theoremstyle{definition}

\newtheorem{lemma}{Lemma}

\usepackage{color}
\usepackage{soul}
\definecolor{mygrey}{gray}{0.35}
\definecolor{myblue}{rgb}{0.2,0.2,0.8}
\definecolor{myzard}{cmyk}{0,0,0.05,0}
\definecolor{mywhite}{rgb}{1,1,1}
\definecolor{myred}{rgb}{0.9,0.1,0.}
\usepackage[colorlinks=true,citecolor=myblue,linkcolor=myblue,urlcolor=myblue]{hyperref}

\usepackage[usenames,dvipsnames]{xcolor} 

\usepackage{lineno}
\usepackage{subfigure}
\usepackage{sublabel}
\usepackage{dcolumn}
\usepackage{amsmath,amssymb}
\usepackage{bm}
\usepackage{bbm}
\usepackage{color}
\usepackage{overpic}
\usepackage{latexsym}
\usepackage{epstopdf}
\usepackage{color}
\usepackage[english]{babel}
\usepackage{latexsym}
\usepackage{stmaryrd}
\usepackage{psfrag,graphicx}
\usepackage{epsf}
\usepackage{subfigure}
\usepackage{amsmath}
\usepackage{amssymb}
\usepackage{amsfonts}
\usepackage{natbib}
\usepackage{epstopdf}
\usepackage{appendix}
\usepackage{enumitem}

\begin{document}

\title{Coherence and non-classicality of quantum Markov processes}

\author{Andrea Smirne}
\affiliation{Institute of Theoretical Physics, Universit{\"a}t Ulm, Albert-Einstein-Allee 11D-89069 Ulm, Germany}

\author{Dario Egloff}
\affiliation{Institute of Theoretical Physics, Universit{\"a}t Ulm, Albert-Einstein-Allee 11D-89069 Ulm, Germany}

\author{Mar{\'i}a Garc{\'i}a D{\'i}az}
\affiliation{Departament de F\'{\i}sica, Grup d'Informaci\'{o} Qu\`{a}ntica, 
Universitat Aut\`{o}noma de Barcelona, ES-08193 Bellaterra (Barcelona), Spain}

\author{Martin B. Plenio}
\affiliation{Institute of Theoretical Physics, Universit{\"a}t Ulm, Albert-Einstein-Allee 11D-89069 Ulm, Germany}

\author{Susana F. Huelga}
\affiliation{Institute of Theoretical Physics, Universit{\"a}t Ulm, Albert-Einstein-Allee 11D-89069 Ulm, Germany}

\begin{abstract}
Although quantum coherence is a basic trait of quantum mechanics, 
the presence of coherences in the quantum description of a certain phenomenon
does not rule out the possibility to give an alternative description of the same phenomenon in purely classical terms.
Here, we give definite criteria to determine when and to what extent quantum coherence is 
equivalent to non-classicality. We prove that a Markovian multi-time statistics
obtained from repeated measurements of a non-degenerate observable cannot be traced back to a classical
statistics if and only if  
the dynamics is able to generate coherences and to subsequently turn them
into populations.
Furthermore, we show with simple examples that such connection between quantum coherence and non-classicality
is generally absent if the statistics is non-Markovian.
\end{abstract}

\maketitle

\section*{Introduction}
The distinction between the classical and the quantum description of 
physical systems has been a central issue from the birth of quantum theory itself \cite{Einstein1935,Schroedinger1935}.
The coherent superposition of states, as well as entanglement are quantum features implying an essentially non-classical
statistics, when proper measurement procedures are devised \cite{Bell1987}, typically involving
the measurement of non-local or non commuting observables. In addition, the non-classical features of a quantum system
can be singled out by means of sequential measurements of one and the same local observable at different times \cite{Leggett1985,Leggett1988,Huelga1995,Lambert2010,Palacios2011,Waldherr2011,Knee2012,Li2012,Kofler2013,Emary2014,Zhou2015,Clemente2016,Friedenberger2017}. 
The evolution of the system between the measurements generally makes the multi-time
statistics highly non-trivial, and a central goal is to relate non-classicality
to easily accessible quantities with a clear physical meaning.

Quantum coherence is a resource, which allows to attain several tasks not achievable 
without it.
Such a basic trait of quantum mechanics has been recently formulated in terms of a resource theory \cite{gour_refframes2008,aaberg2014catalytic,Baumgratz2014,Killoran2016,Yadin2016,Streltsov2017,Liu2017,winter2016operational,chitambar2016critical}.
Within the context of resource theory, classicality is encoded into the notions of incoherent states 
and operations: once a reference basis is fixed, the action of an incoherent operation on an incoherent state
is equivalent to the result of a classical operation. 
At a more practical level, the presence of coherences
in the evolution of a system is often taken in itself as a witness of non-classicality. 
Think, for example, of the intense debate about the possible role of quantum coherence to enhance the efficiency 
of certain biological processes \cite{Engel2007,Ishizaki2009,Chin2013,Huelga2013}. 
The evidence of a coherent coupling between the sites of a molecular complex
certainly challenges the simple classical models based on incoherent transitions among the sites, but it does not rule out the possibility
to explain the observed data via
more elaborate classical descriptions. More in general, the occurrence of coherences in the 
quantum description of a certain phenomenon does not prove by itself 
its non-classical nature \cite{Wilde2010,Briggs2011,Miller2012,Montiel2013,Reilly2014}.

In this paper, we take some relevant steps towards a rigorous link 
between quantum coherence and the non-classicality of multi-time statistics,
identifying proper conditions under  
which such connection can be established unambiguously.
Starting from the quantum description of a system, we exploit a general property of classical stochastic processes,
namely the fulfillment of the Kolmogorov conditions \cite{Feller1971,Breuer2002}, 
to discriminate the multi-time statistics due to repeated projective measurements of one observable from the statistics of a classical process.
This allows us to determine in a precise way when the generation and detection of quantum coherences ``irrevocably excludes" \cite{Wilde2010} alternative, classical explanations.

In particular, we identify the key property of quantum coherences in this context,
and we prove that it is in one-to-one correspondence with the non-classicality of the multi-time statistics,
under the assumption that the latter is Markovian, i.e., that it satisfies the quantum regression theorem
\cite{Lax1968,vanKampen1992,Carmichael1993,Breuer2002,Gardiner2004}.
As a further consequence of our analysis, we illustrate how and to what extent non-classicality can be related with easily detectable quantities,
which can be accessed by carefully preparing the system at the initial time and subsequently measuring it at single instants of time, 
such as those defining the Leggett-Garg type inequalities (LGtIs) \cite{Huelga1995,Waldherr2011,Emary2014,Zhou2015}. 
On the other hand, we also show that when the multi-time statistics is no longer Markovian there 
is no definite 
connection between its non-classicality and the quantum coherences involved in the evolution.


\section*{Multi-time probabilities and classicality}%
Let us first recall the definition of quantum multi-time probability distributions
and the notion of classicality used throughout the work.

Consider a quantum system associated with a Hilbert space $\mathcal{H}$ and evolving unitarily in time.
If we make projective measurements
of the observable 
$
\hat{X} 
$
at the times $t_n \geq \ldots t_1$, with discrete outcomes denoted as $x$,
the joint probability distribution to get
$x_1$ at time $t_1$ and $x_2$ at time $t_2$, $\ldots$ and $x_n$ at time $t_n$ is given by \cite{Gardiner2004}
\begin{equation}\label{eq:hier}
Q^{\hat{X}}_n\left\{x_n, t_n; \ldots x_1,t_1\right\} 
= \mbox{Tr}\left\{\mathcal{P}_{x_n}\mathcal{U}_{t_n-t_{n-1}} \ldots \mathcal{P}_{x_1}\mathcal{U}_{t_1}  \rho(0)\right\},
\end{equation}
where $\mathcal{U}_t\rho = U_t \rho U^{\dag}_t$ 
and $\mathcal{P}_x\rho = \hat{\Pi}_x\rho \hat{\Pi}_x$, with $\hat{\Pi}_x$ projector into the eigenspace of $x$; every super-operator
acts on everything at its right. 
Furthermore, we wrote the operator $\hat{X}$ explicitly to indicate that the statistics will depend on the measured observable.
The collection of joint probability distributions defined in Eq.\eqref{eq:hier} will be the central object
of our analysis; note that, on more mathematical terms, it can be traced back to a proper
definition of quantum stochastic processes, as introduced in \cite{Lindblad1979,Accardi1982} and most recently investigated
by means of the so-called comb formalism \cite{Chiribella2008} in \cite{Milz2017}.

The starting point of our analysis is then the following question: given the quantum 
multi-time probabilities in Eq.\eqref{eq:hier} as input,
how can we certify or exclude that there exists an alternative, classical way to account for them?
The Kolmogorov consistency conditions \cite{Feller1971,Breuer2002}
provide us with a clearcut answer.
In fact, whenever
the probabilities defined in Eq.(\ref{eq:hier}) satisfy
\begin{eqnarray}
&&\sum_{x_k} Q^{\hat{X}}_n\left\{x_n, t_n;  \ldots  x_1,t_1\right\} = Q^{\hat{X}}_{n-1}\left\{x_n, t_n; \ldots \slashed{x_k},\slashed{t_k}  \ldots x_1, t_1\right\} \nonumber\\
&&\forall k \leq n \in \mathbbm{N},\, n>1; \;\; \forall t_n\geq \ldots \geq t_1\in \mathbbm{R}^+;\label{eq:kol}
\;\;\forall x_1, \ldots x_n
\end{eqnarray}
the Kolmogorov extension theorem guarantees the existence of 
a classical stochastic process whose joint probability distributions are equal to these $Q^{\hat{X}}_n$.
Such a process may be rather ad-hoc or exotic, but, as a matter of fact, 
every statement about the quantumness of the outcomes' statistics
and the inherently quantum origin of any related phenomenon
cannot be unambiguously motivated on the basis of probability distributions satisfying Eq.\eqref{eq:kol}. 
Indeed, the joint probabilities of every classical stochastic process do satisfy the Kolmogorov conditions,
while
this is in general not the case for the hierarchy of probabilities in Eq.(\ref{eq:hier}),
since non-selective measurements
(i.e., $\rho \mapsto \sum_x \mathcal{P}_x \rho$) may modify the state of a quantum system.

We can thus formalize the notion of classicality provided by the Kolmogorov conditions, also keeping in mind
that the whole hierarchy of probabilities cannot be reconstructed practically,
as one always deals with a certain finite number of outcomes.\\

\noindent \textbf{Definition 1 (j-classical (jCL) multi-time statistics).}\\
The collection of joint probability distributions $Q^{\hat{X}}_n\left\{x_n, t_n; \ldots x_1,t_1\right\}$
is jCL 
whenever the Kolmogorov conditions in Eq.\eqref{eq:kol} hold for any $n \leq j$;
we say that it is non-classical if it is not even 2CL.\\

Let us stress that identifying the classical statistics with those 
satisfying the Kolmogorov conditions means that we are not taking into
account classical theories with invasive measurements
and, in particular, with singalling in time \cite{Kofler2013,Clemente2016,Halliwell2017},
since the latter would lead to a violation of Eq.\eqref{eq:kol} also at the classical level \cite{Milz2017}.
On the other hand, our definition is certainly well-motivated
by the ubiquity and broad scope of classical stochastic processes.


\section*{Open-quantum-system description and Markovianity}
Before focusing on the possible role of quantum coherence in relation with the notion
of non-classicality specified above,
we want to extend our formalism, by taking into account the
interaction of the measured system with its environment, i.e., treating it as an open quantum system \cite{Breuer2002}.
Indeed, this is to
ensure a realistic description of the system at hand,
including decoherence effects which strongly affect in particular quantum coherence. This
will also allow us to introduce a notion of Markovianity playing a central role in our following analysis.
 
Hence, let us assume that the total system, associated with the Hilbert space $\mathcal{H}$,
is made up of an open system and an environment, i.e.,
we have $\mathcal{H} = \mathcal{H}_S \otimes \mathcal{H}_E$,
$\mathcal{H}_S$ ($\mathcal{H}_E$) being the Hilbert space associated with the system (environment). 
The total system is supposed to be closed, thus evolving via the unitary operators $\mathcal{U}_t$.
Crucially, since the observables we are interested in are
related to the open system only, we focus on measurements of 
observables of the form
$
\hat{X} = \hat{X}_S \otimes \mathbbm{1}.
$
Moreover, $\hat{X}_S$ is assumed to be non-degenerate
and, from now on, $\mathcal{P}_{x}$ denotes a projector defined on $\mathcal{H}_S$ only, $\mathcal{P}_{x} \rho_S = \ketbra{\psi_{x}}{\psi_{x}} \rho_S \ketbra{\psi_{x}}{\psi_{x}}$.

Now, if we assume a product initial state, $\rho(0) = \rho_S(0) \otimes \rho_E(0)$,
with a fixed initial state of the environment,
we can express the one-time statistics of the open system
without referring to the global system and its unitary evolution $\mathcal{U}_t$.
Defining the family of completely positive trace preserving dynamical maps $\left\{\Lambda_S(t)\right\}_{t\in\mathbbm{R^+}}$
via
$$
\rho_S(t) = \Lambda_S(t)\rho_S(0) = \mbox{tr}_E\left\{\mathcal{U}_t \left(\rho_S(0) \otimes \rho_E(0)\right)\right\},
$$
with $\mbox{tr}_E$ ($\mbox{tr}_S$) the partial trace over the environment (system), one has
in fact
$
Q^{\hat{X}_S}_1(x, t) = \mbox{tr}_S\left\{\mathcal{P}_x\Lambda_S(t)\rho_S(0) \right\}.
$
Analogously, the conditional probabilities with respect to the initial time can be written as:
\begin{equation}\label{eq:prec}
Q^{\hat{X}_S}_{1|1}\left\{x,t |x_0, 0\right\} = \mbox{tr}_S\left\{\mathcal{P}_x \Lambda_S(t) \left[\ketbra{\psi_{x_0}}{\psi_{x_0}}\right]\right\}.
\end{equation}

In general, such a simple characterization is not feasible for the higher order statistics:
the multi-time joint probabilities have to be evaluated by referring to the full system, i.e.,
one has to replace $\hat{X} = \hat{X}_S \otimes \mathbbm{1}$
in Eq.(\ref{eq:hier}) and to deal with the whole unitary evolution. Only in this way, in fact, one
can keep track of the correlations between the open system and the environment created 
by their interaction up to a certain time and affecting the open-system multi-time
statistics at subsequent times \cite{Swain1981,Guarnieri2014}.
An important exception to this state of affairs is provided by the quantum regression theorem (QRT) \cite{Lax1968,Swain1981,vanKampen1992,Carmichael1993,Breuer2002,Gardiner2004,Guarnieri2014,Talkner1986,Ford1996,Davies1974,Gorini1976,Ban2017}.
Under proper conditions, which essentially allow to neglect
the effects of system-environment correlations at a generic time \cite{Swain1981,Carmichael1993,Guarnieri2014}, 
the joint distributions $Q^{\hat{X}_S}_n$
can be fully determined by the initial reduced state $\rho_S(0)$ and the dynamical maps $\Lambda_S(t)$.
In the following, whenever we assume the QRT, we also assume that the system dynamics 
is described by the Lindblad equation \cite{Lindblad1976, Breuer2002}
$ d\rho_S(t)/dt = \mathcal{L}\rho_S(t) = -i \left[\hat{H}, \rho_S(t)\right] + \sum_k \left(
\hat{L}_k \rho_S(t) \hat{L}^{\dag}_k - \frac{1}{2}\left\{\hat{L}^{\dag}_k \hat{L}_k,\rho_S(t) \right\}\right)$, with
$\hat{H}^{\dag} = \hat{H}$ and $\hat{L}_k$ linear operators on $\mathcal{H}_S$;
the corresponding dynamical
maps
$
\Lambda_S(t) = e^{\mathcal{L}t}
$
satisfy the semigroup composition law $\Lambda_S(t)\Lambda_S(s) = \Lambda_S(t+s)$ for any $t,s \in\mathbbm{R^+}$.
Explicitly, the QRT for the joint probability distributions associated with projective measurements implies
\begin{widetext}
\begin{eqnarray}\label{eq:LQRT}
 Q^{\hat{X}_S}_n\left\{x_n, t_n \ldots x_1,t_1\right\}  = \mbox{tr}_S\left\{\mathcal{P}_{x_n}e^{\mathcal{L}(t_n-t_{n-1})}\ldots \mathcal{P}_{x_1}e^{\mathcal{L}t_1}\rho_S(0)\right\}. 
\end{eqnarray}
\end{widetext}
The previous relation is similar to the general definition in Eq.(\ref{eq:hier}), but, crucially,
now the whole hierarchy of probabilities involves exclusively objects referring to the open system only.

As shown in App.\ref{app:qm}, the QRT for a non-degenerate observable implies the following
property of the conditional probabilities:
\begin{align}
&Q^{\hat{X}_S}_{1|n}\left\{x_{n+1}, t_{n+1}| x_{n}, t_{n}; \ldots x_1,t_1\right\} \nonumber\\
&= Q^{\hat{X}_S}_{1|1}\left\{x_{n+1}, t_{n+1}| x_{n}, t_{n}\right\},  \nonumber
\end{align}
which is nothing else than the Markov condition \cite{Feller1971,Breuer2002}
stating that the value of the observable at a certain time, conditioned
on its previous history, only depends on the last assumed value.
As well-known,
the QRT plays the counterpart of classical Markov processes for the quantum multi-time statistics \cite{Lindblad1979,Fleming2011,LoGullo2014}; 
see App.\ref{app:qm}, also in relation with the notions of quantum Markovianity
referring, instead, to the dynamics \cite{Rivas2014, Breuer2016}.
We then proceed by introducing the following definition.\\

\noindent \textbf{Definition 2 (j-Markovian (jM) multi-time statistics).}
The collection of joint probability distributions $Q^{\hat{X}_S}_n\left\{x_n, t_n; \ldots x_1,t_1\right\}$ is jM \cite{foot4}
if it can be written as in Eq.\eqref{eq:LQRT}
for any $n\leq j, x_1, \ldots x_n, t_n \geq \ldots t_1\in \mathbbm{R}^+$; it is non-Markovian (NM) if it is not even 2M.\\

The key property of Markovian processes (irrespective of whether there is or there is not
an equivalent classical description of them) is that the entire hierarchy of probabilities can be reconstructed from the initial probability
$Q^{\hat{X}_S}_1\left\{x_0,0\right\}$ and the transition probabilities $Q^{\hat{X}_S}_{1|1}\left\{x, t|y, s\right\}$.
As we will see, this plays a basic role in our analysis.



\section*{Generating and detecting quantum coherence} 
Here we present the property of quantum coherence directly related to the non-classicality possibly emerging from 
repeated measurements of a quantum observable. 
Roughly speaking, we need to characterize the evolutions which not only generate coherences, but can also turn such coherences
into the populations measured at a later time. 

Therefore, consider the following definition, which refers explicitly to Lindblad dynamics; 
in App.\ref{app:ex} we introduce the definition for a generic (divisible) dynamics.\\

\noindent \textbf{Definition 3 (Coherence-generating-and-detecting (CGD) dynamics).}
The Lindblad dynamics \cite{foot5} $\left\{\Lambda(t)=e^{\mathcal{L} t}\right\}_{t\in\mathbbm{R^+}}$
        is CGD
	whenever there exist $t, \tau \in\mathbbm{R^+}$ such that
	(here, for the sake of clarity, we denote explicitly the map composition as $\circ$)
	\begin{equation}\label{eq:dellam}
	\Delta \circ \Lambda(t)\circ \Delta \circ \Lambda(\tau)\circ\Delta \neq \Delta \circ \Lambda(t+\tau)\circ\Delta,
	 \end{equation}
	where $\Delta= \sum_x \mathcal{P}_x$ is the complete dephasing map; otherwise, the dynamics is denoted as NCGD.\\
		
We always assume that the reference basis defining $\Delta$ coincides with the eigenbasis
of the measured observable $\hat{X}$.
See Fig.\ref{fig:1} for an illustrative sketch of the NCGD definition.

\begin{figure}[ht!]
\begin{center}
\includegraphics[width=0.8\columnwidth]{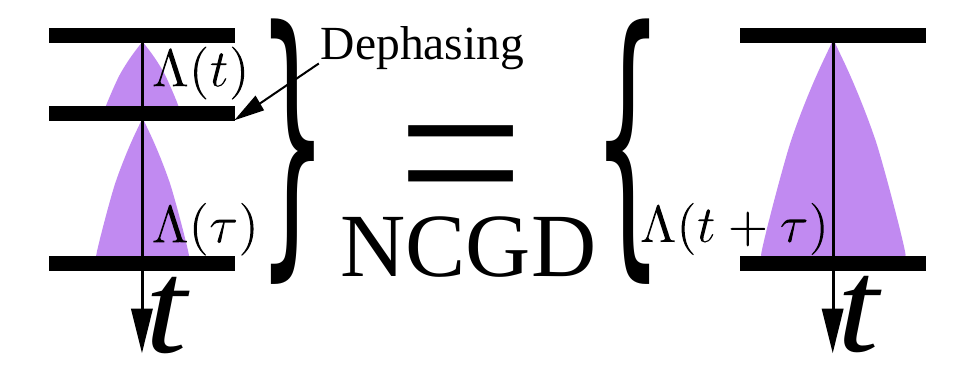}
\caption{Illustrative sketch of the property of NCGD dynamics: 
fixed a reference basis, applying dephasing at an intermediate 
instant of the dynamics does not change the state transformation, if dephasing is also
applied at the initial and final time.}
\label{fig:1}
\end{center}
\end{figure}

Recently there has been considerable interest in incoherent operations
\cite{aaberg2014catalytic,Baumgratz2014, Streltsov2017,Liu2017}, which are defined for maps only. 
In contrast here we investigate dynamics. We can compare our notion with the literature by fixing $t= \tau$
in the Definition 3, thus referring it to one map, $\Lambda \equiv \Lambda(t)$, which we call CGD map.
There are two interesting subsets of NCGD maps.
One is the subset that does not create coherence  from incoherent states, which is described by $\Delta \circ \Lambda\circ\Delta = \Lambda\circ\Delta$; this is the maximal set of incoherent operations~\cite{aaberg2014catalytic}. 
The other noteworthy subset of NCGD maps is the coherence nonactivating set
fixed by $\Delta \circ \Lambda\circ\Delta = \Delta\circ\Lambda$; here, since the populations are independent of the initial coherence, the coherence is not a useable resource~\cite{Liu2017}. 
Operations that are neither incoherent nor coherence nonactivating may still be NCGD, if the
sub\-spaces where coherence is generated are different from the ones detecting it (see App.~\ref{app:ch} for a detailed example). 

We conclude that NCGD dynamics can be understood by the propagated population not depending on the generated coherences.
In addition, we can provide a direct operational meaning to (N)CGD, as ensured by the following proposition, which is proved in App:~\ref{app:semi}.\\

\noindent \textbf{Proposition 1.}
Given a non-degenerate reduced observable 
$\hat{X} = \sum_{x} x \ketbra{\psi_x}{\psi_x}$ and the Lindblad dynamics $\left\{\Lambda(t) = e^{\mathcal{L}t}\right\}_{t \in\mathbbm{R^+}}$, 
the latter is NCGD if and only if 
the conditional probabilities $Q^{\hat{X}}_{1|1}\left\{x,t|x_0,0\right\}$
satisfy $\forall t \geq s \in\mathbbm{R^+}$
\begin{equation}\label{eq:chkol}
Q^{\hat{X}}_{1|1}\left\{x,t|x_0,0\right\} = \sum_{y}   Q^{\hat{X}}_{1|1}(x, t-s | y, 0)  Q^{\hat{X}}_{1|1}(y, s | x_0, 0).
\end{equation}

The condition in Eq.\eqref{eq:chkol} is simply the (homogeneous) Chapman-Kolmogorov equation \cite{vanKampen1992,Breuer2002,Gardiner2004},
which is always satisfied by a classical Markov (homogeneous) process, but, indeed, not necessarily by a quantum one.
As we will see, the relation between NCGD and classicality relies on Proposition 1.
For the moment, let us stress that Eq.\eqref{eq:chkol} can be in principle easily checked in practice, since 
the conditional probabilities $Q^{\hat{X}}_{1|1}\left\{x,t|x_0,0\right\}$ can be reconstructed by preparing the system in one eigenstate of $\hat{X}$
and measuring $\hat{X}$ itself at a final time $t$, without the need to access intermediate steps
of the evolution. 

The previous proposition also allows us to connect CGD with 
other easily accessible quantities, which are already well-known in the literature. As a significant example, let us mention
the LGtIs \cite{Huelga1995,Emary2014},
which were introduced to characterize macroscopic realistic theories,
that is, classical theories where physical systems
possess definite values (whether or not they are measured)
and such values can be accessed without changing the state of the system.
In particular, in LGtIs the 
Leggett-Garg non-invasiveness requirement \cite{Leggett1985}
is replaced by an assumption which turns out to be related to Markovianity \cite{Emary2014}.
Given a dichotomic observable $X$ with values in $\left\{0,1\right\}$ and the related correlation function, which in the quantum framework is defined as
$C_X(t,0) = \sum_{x, x_0} Q^{\hat{X}}_2\left\{ x, t;  x_0, 0\right\} x  x_0$, the LGtI
we consider here reads
$
\left |2C_X(t,0) - C_X(2t,0)\right |  \leq  \langle X(0) \rangle$, 
with $\langle X(0) \rangle$ the expectation value of $X$ at the initial time.
Now, since the validity of Eq.\eqref{eq:chkol} is a sufficient condition for the LGtI to be satisfied, see App.\ref{app:semi}, 
Proposition 1 directly leads us to the following.\\

\noindent \textbf{Theorem 1.}
Given a Lindblad dynamics, 
the LGtI is violated only if the dynamics is CGD.\\

The Theorem thus clarifies how the LGtI can be used to witness
that coherences are generated by the dynamics 
and subsequently turned into populations. 

\section*{Quantum coherence and non-classicality}
We are now in the position to state the main result of our paper.
In the previous section we have seen how the
NGDC property of the dynamics is related with the Chapman-Kolmogorov composition law
for the conditional probabilities with respect to the initial time [see Proposition 1].
However,
if we want to establish a definite connection between coherences and (non-)classicality, we need 
to take a step further and to go beyond the one-time statistics
to access the higher orders
of the hierarchy of probabilities, since only the latter encompass the definite meaning of classicality we are referring to. 

The recalled notion of quantum
Markovianity for multi-time statistics does provide us with the wanted link among coherences and classicality. 
This is shown by the following Theorem,
whose proof is presented in App.\ref{app:semi2}.\\

\noindent \textbf{Theorem 2.}
Given a non-degenerate reduced observable 
$\hat{X} = \sum_{x} x \ketbra{\psi_x}{\psi_x}$ and a jM hierarchy 
of probabilities $Q^{\hat{X}}_n\left\{x_n, t_n; \ldots x_1,t_1\right\}$,
the latter is jCL for any initial diagonal state $\rho(0) = \sum_{x_0} p_{x_0} \ketbra{\psi_{x_0}}{\psi_{x_0}}$ 
if and only if the dynamics
$\left\{\Lambda(t) = e^{\mathcal{L}t}\right\}_{t\in \mathbbm{R}^+}$
is NCGD.\\

Theorem 2 means that if the multi-time statistics is Markovian, the capability of a dynamics to generate coherences and turn them into populations
is in one-to-one correspondence with non-classicality.
In other words, Markovianity guarantees the wanted connection between a property of the coherences, which is
fixed by the dynamics, and the classicality of the multi-time probability distributions.
This is a direct consequence of the peculiarity of Markovian processes, classical as well as quantum,
which allows one to reconstruct the higher order probability distributions from the lowest order one.

Finally, the previous Theorem also allows us to clarify to what extent the LGtI is actually related with non-classicality,
since it directly implies the following.\\

\noindent \textbf{Theorem 3.}
Given a 2M hierarchy 
of probabilities, the LGtI is violated only if the hierarchy is non-classical.

For the sake of clarity, in Fig.\ref{fig:1a} we report a summary of the theorems presented
in this paper, which establish definite connections among the notions
of classicality, quantum coherence (in particular NCGD of the dynamics)
and LGtI.

\begin{figure}[ht!]
\begin{center}
\includegraphics[width=0.75\columnwidth]{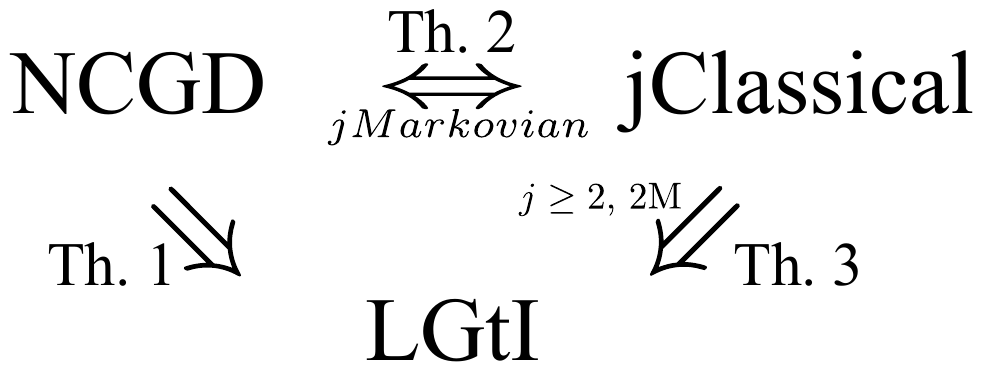}
\caption{Implication structure of the main results of the paper. 
The notion of $j$-classicality is given in Definition 1, $j$-Markovianity in Definition 2
and the property of the evolution of coherence 
named NCGD, which stands for non-coherence-generating-and-detecting, in Definition 3;
finally LGtI denotes the Leggett-Garg type inequality \cite{Huelga1995} considered here.  
A Lindblad dynamics is always assumed.}
\label{fig:1a}
\end{center}
\end{figure}

\section*{Non-Markovian multi-time statistics}
In the last part of the paper, we start to explore the general case
of non-Markovian multi-time statistics.
Indeed, now the connection between quantum coherence
and the non-classicality of the multi-time statistics is no longer guaranteed,
since the higher order probabilities cannot be inferred from the lowest order ones. 
Exploiting
a model which traces back to Lindblad himself \cite{Lindblad1980, Accardi1982},
we show that this connection is lost even in the presence of a simple Lindblad dynamics. 
Remarkably, there can be a genuinely non-classical statistics associated with the measurements
of an observable without that any quantum coherence of such observable is ever present in the state of the measured system.

Consider a two-level system, $\mathcal{H}_S= \mathbbm{C}^2$, interacting with a continuous degree of freedom, $\mathcal{H}_E = \mathcal{L}(\mathbbm{R})$,
via the unitary operators defined 
by
$
U(t) \ket{\ell, p} = e^{i \ell p t} \ket{\ell, p},
$
where $\left\{\ket{\ell}\right\}_{\ell=-1,1}$ is the eigenbasis of the system operator $\hat{\sigma}_z$
and 
$\left\{\ket{p}\right\}_{p \in \mathbbm{R}}$ is an improper basis of $\mathcal{H}_E$. 
Assuming an initial product state
and a pure environmental state, $\rho_E(0) = \ketbra{\varphi_{E}}{\varphi_{E}}$ with
$
\ket{\varphi_{E}} = \int_{-\infty}^{\infty} d p f(p) \ket{p}$,
the open-system dynamics is pure dephasing, fixed by $\rho_{-11}(t)=\rho_{-11}(0) k(t)$
with $k(t) = \int_{-\infty}^{\infty} d p |f(p)|^2 e^{2 i p t}$, where $\rho_{-11} = \bra{-1} \rho \ket{1}$.
We consider projective measurements of $\hat{\sigma}_x$, whose eigenbasis is denoted as $\left\{\ket{+}, \ket{-}\right\}$,
and then we assume initial states as
$\rho(0) = p_+ \ketbra{+}{+}+p_- \ketbra{-}{-}$.
In App.\ref{app:ex} we report the exact two-time
probability $Q^{\hat{\sigma}_x}_2$, given by Eq.\eqref{eq:hier},
and the probability $Q^{\hat{\sigma}_x}_{2M}$ one would get for a Markovian statistics,
see Eq.\eqref{eq:LQRT},
along with the conditions for the dynamics to be CGD and the statistics 2CL.

First, consider an initial Lorentzian distribution,
$
|f(p)|^2 = \Gamma \pi^{-1}/(\Gamma^2+(p-p_0)^2),
$
so that the decoherence function is given exactly by an exponential,
$
k(t)= e^{2 i p_0 t} e^{- 2 \Gamma t},
$
and the open-system dynamics is fixed by the pure dephasing Lindblad equation \cite{Guarnieri2014}.
Nevertheless, the QRT in Eq.\eqref{eq:LQRT} is generally not satisfied, 
not even by the two-time probability distributions, so that the multi-time statistics is NM.
The difference, also qualitative, among the exact joint probability distribution $Q^{\hat{\sigma}_x}_2$
and the Markovian one $Q^{\hat{\sigma}_x}_{2M}$ is illustrated in Fig.\ref{fig:2}  {\bf a)}.
Furthermore, one can easily see App.\ref{app:ex} that the Kolmogorov condition for $n=2$
does not hold, $\sum_y Q^{\hat{\sigma}_x}_{2}\left\{+, t; y, s\right\} \neq Q^{\hat{\sigma}_x}_{1}\left\{+, t\right\}$.
The statistics at hand is hence non-classical. On the other hand, the corresponding Lindblad dynamics (for $p_0=0$) is NCGD:
pure dephasing on $\hat{\sigma}_z$ cannot even generate coherences of $\hat{\sigma}_x$ (of course,
this also implies that the corresponding LGtI is always satisfied). 
We conclude that, despite the non-classicality of the statistics,
the coherences of the measured observable $\hat{\sigma}_x$ are not involved at all in the dynamics:
at no point in time the state of the measured system has non-zero coherence
with respect to $\hat{\sigma}_x$. 

\begin{figure}[ht!]
\begin{center}
\includegraphics[width=0.45\columnwidth]{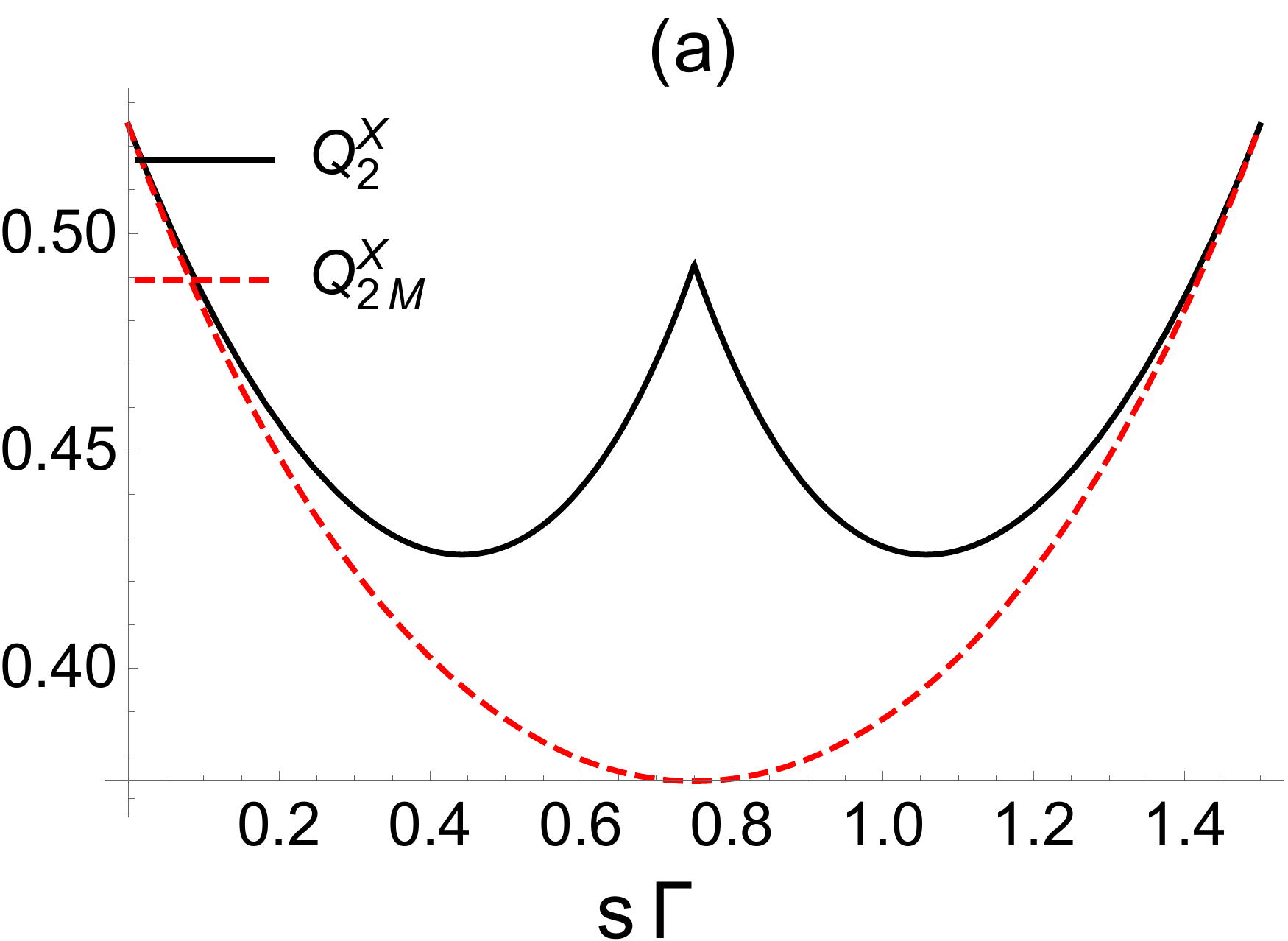} \includegraphics[width=0.48\columnwidth]{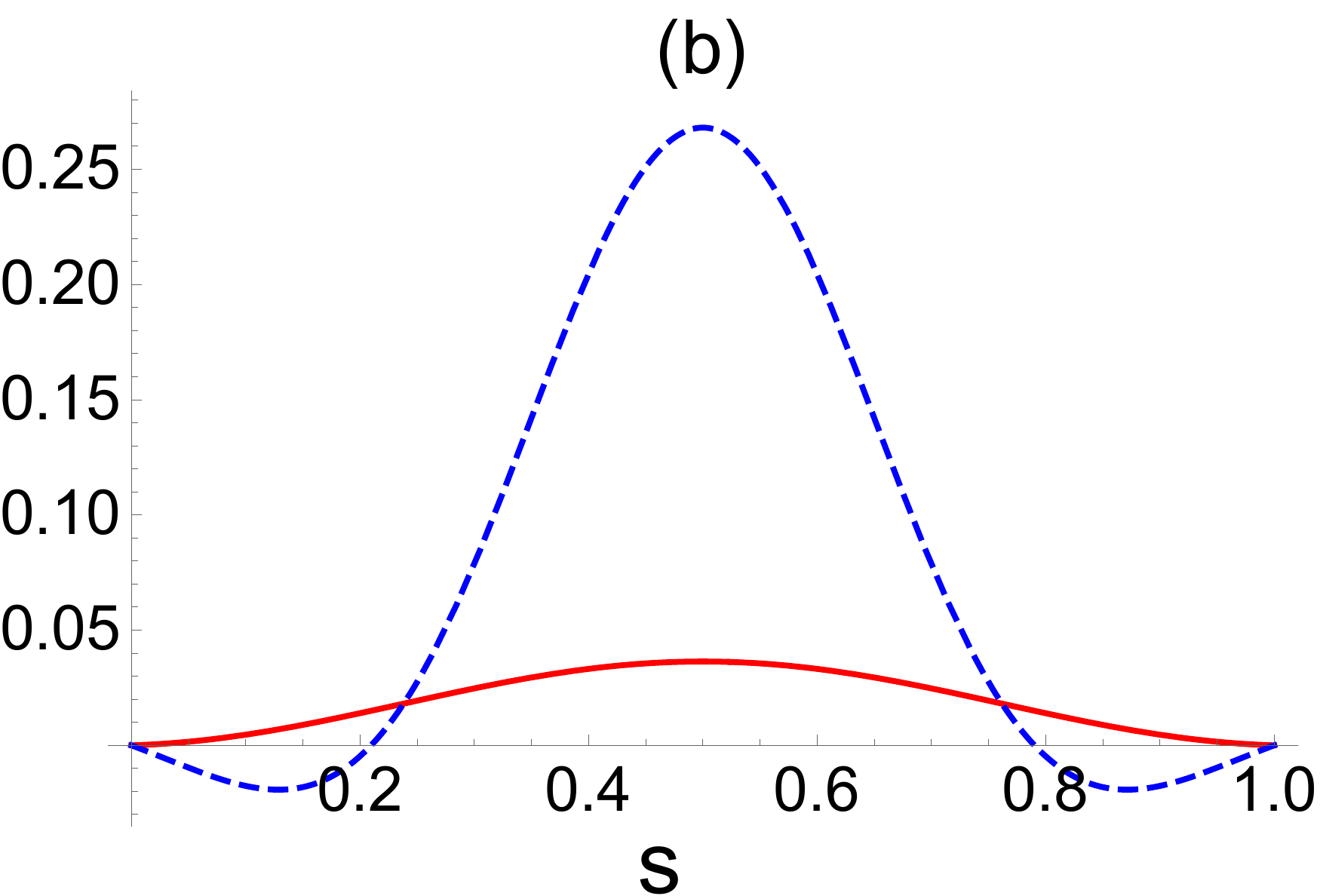}
\caption{{\bf a)} Exact two-time probability
$Q^{\hat{\sigma}_x}_{2}\left\{+, t; +, s\right\}$ [see Eq.\eqref{eq:hier}] (solid black line) 
and 2M probability $Q^{\hat{\sigma}_x}_{2M}\left\{+, t; +, s\right\}$ [see Eq.\eqref{eq:LQRT}] (dashed red line)
as fun\-ctions of $s$,
for a Lorentzian $|f(p)|^2$.
The parameters are $t=1.5\Gamma^{-1}, p_0=0$.
 {\bf b)} Amount of CGD (according to the general definition given in App.\ref{app:ex}) quantified via the
infinity norm, (solid, red line) and violation of the Kolmogorov condition [i.e., 
 $|\sum_y Q^{\hat{\sigma}_x}_{2}\left\{+, t; y, s\right\} - Q^{\hat{\sigma}_x}_{1}\left\{+, t\right\}|$] (dashed, blue line) as
 functions of $s$,
for $|f(p)|^2$ given by the superposition of two Gaussians.
The parameters are $A_1=A_2=1/(2\sqrt{2\pi}), \sigma_1=\sigma_2=\sigma, \sigma=p_1=t=1, p_2=2p_1$.
}
\label{fig:2}
\end{center}
\end{figure}

In a complementary way, we exemplify how the instants 
where the multi-time statistics satisfies the Kolmogorov conditions may coincide with instants where 
coherences are generated and converted into populations. 
However, we have to leave open the question of whether 
there is a fully classical statistics (for any sequence of times), while the dynamics of the coherences is non-trivial.
Take an initial distribution given by the sum of two Gaussians, $|f(p)|^2= \sum_{i=1,2} A_i \exp\left(-(p-p_i)^2/(2 \sigma^2_i)\right)$.
Once again one can easily see that the statistics is NM [see App.\ref{app:ex}], but this time the dynamics is generally CGD. 
Nevertheless, the creation and detection of coherences is not in correspondence with the non-classicality of the statistics.
In Fig.\ref{fig:2} {\bf b)}, we can see that there are instants of time where the dynamics is CGD, but the statistics is 2CL, see also App.\ref{app:ex}; note that at these same instants of time also the Chapman-Kolmogorov condition in
Eq.\eqref{eq:chkol} does not hold.
By investigation (not reported here) of the model at hand in a wide region of parameters, we also observe that there does not seem to be a threshold in the amount of CGD, above which the violation of 2CL is guaranteed.

The previous examples illustrate the essential role of Markovianity to establish a precise link between quantum coherence, 
or any other dynamical property,
and non-classicality. In addition, they imply that
the coherences themselves cannot be used as a witness of non classicality, without any a-priori information about higher order probabilities. 
To know whether coherences are linked to jCL, one needs to access $Q^{\hat{X}}_j$ to verify jM,
but then jCL can be directly checked via the Kolmogorov conditions.

\section*{Conclusions}
We proved a one-to-one correspondence between 
the non-classicality of the multi-time statistics
associated with sequential measurements of one observable at different
times and the quantum coherence with respect to the eigenbasis of the measured observable itself.
We pointed out the key property of quantum coherence which is directly linked with non classicality,
connecting it to the recently developed resource theory of coherence.
Furthermore, we illustrated the essential role of Markovianity in linking dynamical properties, such as the evolution of quantum coherence or the violation of the LGtI,
to higher order probability distributions of the multi-time statistics and hence to (non-) classicality.

Our approach will naturally apply to several areas where the possible
quantum origin of certain physical phenomena is under debate, such as quantum biology or quantum thermodynamics.
We plan to exploit the results presented here to study some relevant examples taken from these fields of research.
In addition, we want to extend our analysis to classical theories with invasive measurements
starting from the notion of signaling in time; in particular,
we think that our approach will further motivate the investigation of the scenario in which memory effects
are present, at the level of the quantum multi-time statistics
and/or of the measurement invasiveness  \cite{Montina2012,Frustaglia2016,Cabello2018}.
Finally, it will be of interest to include the description of non projective measurements, as well as the
measurement of multiple, non-commuting observables. \\

{\bf{Acknowledgments}}
	We acknowledge very interesting and fruitful discussions with Thomas Theurer and Benjamin Desef,
	as well as financial support by the ERC Synergy Grant BioQ (grant no 319130),
	the EU project QUCHIP (grant no. 641039) and Fundaci{\'o}n {''}la Caixa{''}.

\newpage

\onecolumngrid
\appendix

\section{Quantum regression theorem and quantum Markovianity}\label{app:qm}
Here we want to discuss more in detail why the QRT can be naturally seen as the quantum counterpart 
of the Markov condition for the hierarchy of probabilities defined in the main text.

Now, let 
\begin{eqnarray}\label{eq:cond}
 Q^{\hat{X}}_{k|n} \left\{x_{n+k}, t_{n+k}; \ldots x_{n+1}, t_{n+1}| x_{n}, t_{n}; \ldots x_1, t_1  \right\}
= \frac{Q^{\hat{X}}_{k+n} \left\{x_{n+k}, t_{n+k}; \ldots x_1, t_1  \right\}}{Q^{\hat{X}}_{n} \left\{x_{n}, t_{n}; \ldots x_1, t_1 \right\}} 
\end{eqnarray}
be the conditional probability distributions associated with the hierarchy of probability distributions defined in Eq.(1) of the main text, referred to a reduced observable $\hat{X}\otimes \mathbbm{1}$
(we keep implying the subfix $S$ for the sake of simplicity); Eq.\eqref{eq:cond} can be easily obtained from the general definition
based on the conditional state \cite{Gardiner2004}, using the projectors into the eigenspaces of $\hat{X}\otimes \mathbbm{1}$, which are degenerate
with respect to the global space $\mathcal{H}$.
If we now express the right hand side of the previous relation via the QRT, i.e., Eq.(4) of the main text, we can exploit the non-degeneracy of $\hat{X}$ on $\mathcal{H}_S$.
Because of that, for any couple of reduced super-operators $\mathcal{A}$ and $\mathcal{B}$
we have
$$
\frac{ \mbox{tr}_S\left\{\mathcal{A} \mathcal{P}_x \rho_S\right\}}{ \mbox{tr}_S\left\{\mathcal{B} \mathcal{P}_x \rho_S\right\}} 
= \frac{ \mbox{tr}_S\left\{\mathcal{A} [\ketbra{\psi_x}{\psi_x}]\right\}}{ \mbox{tr}_S\left\{\mathcal{B} [\ketbra{\psi_x}{\psi_x}]\right\}},
$$
from which one can easily see that the hierarchy defined in Eq.(4) of the main text satisfies the
condition
\begin{eqnarray}
 &&Q^{\hat{X}_S}_{1|n}\left\{x_{n+1}, t_{n+1}| x_{n}, t_{n}; \ldots x_1,t_1\right\} = Q^{\hat{X}_S}_{1|1}\left\{x_{n+1}, t_{n+1}| x_{n}, t_{n}\right\}  \nonumber\\
&&= \mbox{tr}_S\left\{\mathcal{P}_{x_{n+1}} e^{\mathcal{L}(t_{n+1}-t_{n})}[\ketbra{\psi_{x_{n}}}{\psi_{x_{n}}}]\right\}.\label{eq:qmarkov}
\end{eqnarray}
The first equality in the previous equation is the Markov condition, which defines Markov stochastic processes \cite{Breuer2002,Gardiner2004}.
Actually, the QRT is at the basis of the definition of quantum Markov processes put forward by Lindblad in \cite{Lindblad1979}
(see also the more recent definitions in \cite{Fleming2011, LoGullo2014}). 

On the other hand, different approaches have been followed to introduce a definition
of quantum Markovianity which can be referred to the dynamics of the open quantum systems \textit{tout-court}.
Indeed, the hierarchy of probabilities in Eq.(1), or in Eq.(4), of the main text
depends on the specific measurement procedure one is taking into account. Moreover and more importantly, the measurements
at intermediate times involved in that definition modify the correlations between the system and the environment, thus
modifying the subsequent dynamics of the open system \cite{Breuer2016}.
Hence, in order to assign the Markovian or non-Markovian attribute to the open system \textit{dynamics solely}, 
different definitions have been put forward (see the recent reviews \cite{Rivas2014, Breuer2016}),
relying directly on properties of the dynamical maps $\left\{\Lambda(t)\right\}_{t \in \mathbbm{R}^+}$.
Of course, the Markovianity referring to the dynamics solely and that referring to the multi-time probability distributions are quite different concepts, 
since in general the one-time statistics (and the related dynamical maps) does not allow to infer
the behavior of higher order distributions \cite{Vacchini2011}, and then, e.g., whether the QRT holds or not. 
The `proper' definition of quantum Markovianity ultimately depends on the framework one is interested in. 
Here, as said in Definition 2, we identify quantum Markov processes with those satisfying the QRT.

Finally, let us note that the Markov condition in Eq.(\ref{eq:qmarkov}) implies that, also in the
quantum case, the whole hierarchy of probabilities is fixed by the initial condition and the $Q^{\hat{X}_S}_{1|1}$
conditional probabilities,
i.e., we have for any $x_1, \ldots x_n, t_n \geq \ldots \geq t_1$
\begin{eqnarray}
 Q^{\hat{X}_S}_n\left\{x_n, t_n; \ldots; x_2, t_2; x_1, t_1\right\} 
&&=\left(
\prod_{k=1}^{n-1}Q^{\hat{X}_S}_{1|1}\left\{x_{k+1}, t_{k+1}| x_k, t_k\right\}
\right)
Q^{\hat{X}_S}_{1}\left\{ x_1, t_1\right\} \label{eq:eq} \\
&&= \left(
\prod_{k=1}^{n-1} \mbox{tr}\left\{\mathcal{P}_{x_{k+1}} e^{\mathcal{L}(t_{k+1}- t_k)} \left[\ketbra{\psi_{x_k}}{\psi_{x_k}}\right]\right\}
\right) 
\mbox{tr}\left\{\mathcal{P}_{x_{1}} e^{\mathcal{L} t_1} \rho_S(0)\right\}, \nonumber
\end{eqnarray}
where in the second line we used that since we are focusing on the Lindblad dynamics we have
\begin{eqnarray}\label{eq:qhat}
&&Q^{\hat{X}_S}_{1|1}\left\{x, t| y, s \right\} = \mbox{tr}_S\left\{\mathcal{P}_x e^{\mathcal{L}(t-s)} \left[\ketbra{\psi_{y}}{\psi_{y}}\right]\right\}
= Q^{\hat{X}_S}_{1|1}\left\{x , t-s| y, 0\right\}.
\end{eqnarray}

\section{Proof of the Proposition 1 and Theorem 1.}\label{app:semi}

Before proving the Proposition 1, we need to prove the following Lemma.

\begin{lemma}
The evolution fixed by the Lindblad dynamics $\left\{\Lambda(t)=e^{\mathcal{L}t}\right\}_{t \in \mathbbm{R}^+}$ is NCGD if and only if 
\begin{eqnarray}
 \sum_{y \neq z}\bra{\psi_{\tilde{x}}}\Lambda(t)[\ketbra{\psi_y}{\psi_{z}}]\ket{\psi_{\tilde{x}}} \bra{\psi_y}\Lambda(\tau)\left[\ketbra{\psi_x}{\psi_x}\right]\ket{\psi_{z}} = 0
\quad \forall x, \tilde{x}; t, \tau\in \mathbbm{R}^+. \label{eq:cop2}
\end{eqnarray}
\end{lemma}

\begin{proof}
	First note that $\left\{\Lambda(t)\right\}_{t \in \mathbbm{R}^+}$ satisfies Eq.\eqref{eq:cop2} if and only if $\forall x,\tilde{x}; t, s \in \mathbbm{R}^+$
	\begin{eqnarray}
	&& \bra{\psi_{\tilde{x}}}\;\Lambda(t)\circ \Lambda(\tau)\left[\ketbra{\psi_{x}}{\psi_{x}} \right]  \;\ket{\psi_{\tilde{x}}} \nonumber\\
	&&=\sum_{y,z} \bra{\psi_{\tilde{x}}}\;\Lambda(t)\Big[ \ketbra{\psi_y}{\psi_y} \; \Lambda(\tau)\left[\ketbra{\psi_{x}}{\psi_{x}} \right] \;\ketbra{\psi_z}{\psi_z}  \Big]  \;\ket{\psi_{\tilde{x}}} \nonumber\\
	&&=
	\sum_{y} \bra{\psi_{\tilde{x}}}\;\Lambda(t)\Big[ \ketbra{\psi_y}{\psi_y} \; \Lambda(\tau)\left[\ketbra{\psi_{x}}{\psi_{x}} \right] \;\ketbra{\psi_y}{\psi_y}  \Big]  \;\ket{\psi_{\tilde{x}}} \nonumber\\
	&&=
	\bra{\psi_{\tilde{x}}}\;\Lambda(t)\circ\Delta \circ\Lambda(\tau)\left[\ketbra{\psi_{x}}{\psi_{x}} \right]  \;\ket{\psi_{\tilde{x}}},\label{eq:innn}
	\end{eqnarray}
	since the first and the last equalities are always valid, while the second is Eq.\eqref{eq:cop2}, adding the diagonal terms on both sides.
	We first prove that if $\left\{\Lambda(t)\right\}_{t \in \mathbbm{R}^+}$ is NCGD then the above equality holds.
	We can in fact rewrite the first line as
	\begin{eqnarray}	
	&&\sum_{k,k'} \bra{\psi_{\tilde{x}}} \ketbra{\psi_k}{\psi_k}\;
	\Lambda(t)\circ \Lambda(\tau)\big[
	\ketbra{\psi_{k'}}{\psi_{k'}} \ketbra{\psi_{x}}{\psi_{x}} \ketbra{\psi_{k'}}{\psi_{k'}} \big]  \;
	\ketbra{\psi_k}{\psi_k}\ket{\psi_{\tilde{x}}} \nonumber \\
	&&= \bra{\psi_{\tilde{x}}} \;
	\Delta\circ \Lambda(t)\circ \Lambda(\tau)\circ\Delta\left[
	\ketbra{\psi_{x}}{\psi_{x}} \right]   \; \ket{\psi_{\tilde{x}}}
	\\
	&&=\bra{\psi_{\tilde{x}}} \;
	\Delta\circ \Lambda(t)\circ\Delta\circ \Lambda(\tau)\circ\Delta\left[
	\ketbra{\psi_{x}}{\psi_{x}} \right]   \; \ket{\psi_{\tilde{x}}} \nonumber
	\\
&&=
\bra{\psi_{\tilde{x}}}\;\Lambda(t)\circ\Delta \circ\Lambda(\tau)\left[\ketbra{\psi_{x}}{\psi_{x}} \right]  \;\ket{\psi_{\tilde{x}}}. \nonumber
	\end{eqnarray}
	where in the first line we used that only the terms $\tilde{x}=k$,$x=k'$ are non-zero, and in the third line we used that 
	$\left\{\Lambda(t)\right\}_{t\in \mathbbm{R}^+}$ is NCGD.
	
	For the converse, we start with the assumption that the equality \eqref{eq:innn} holds for any $x,\tilde{x}, t, \tau \in \mathbbm{R}^+$. The statement then simply follows by the linearity of the propagators since $\Delta$ is a projection onto the span of $\{\ketbra{\psi_x}{\psi_x}\}_{x}$.
\end{proof}


We can now prove the Proposition 1.
\begin{proof} 
Using Eq.(3) of the main text and Eq.\eqref{eq:qhat} we have that
\begin{eqnarray}
 && Q^{\hat{X}}_{1|1}\left\{x,t|x_0,0\right\} - \sum_{y}   Q^{\hat{X}}_{1|1}(x, t-s | y, 0)  Q^{\hat{X}}_{1|1}(y, s | x_0, 0)  \nonumber\\
 &&=\bra{\psi_x}\Lambda(t) \left[\ketbra{\psi_{x_0}}{\psi_{x_0}}\right] \ket{\psi_x} - \sum_{y} \bra{\psi_x} \Lambda(t-s) \left[\ketbra{\psi_{y}}{\psi_{y}}\right] \ket{\psi_x}
* \bra{\psi_{y}} \Lambda(s) \left[\ketbra{\psi_{x_0}}{\psi_{x_0}}\right] \ket{\psi_{y}},  \nonumber
\end{eqnarray}
so that using the semigroup composition law $\Lambda(t)=\Lambda(t-s)\Lambda(s)$ and the resolution of the identity,
the first term in the previous expression can be written as
\begin{eqnarray}
&&\bra{\psi_x} \Lambda(t-s)\big[\Lambda(s)  \left[\ketbra{\psi_{x_0}}{\psi_{x_0}}\right] \big] \ket{\psi_x}  \nonumber\\
&&= \sum_{y, y'}\bra{\psi_x} \Lambda(t-s)\Big[\ketbra{\psi_{y}}{\psi_{y}}
\Lambda(s) \left[\ketbra{\psi_{x_0}}{\psi_{x_0}}\right] \ketbra{\psi_{y'}}{\psi_{y'}}\Big]\ket{\psi_x},  \nonumber\\
&&=  \sum_{y, y'}\bra{\psi_x} \Lambda(t-s)\left[\ket{\psi_{y}}\bra{\psi_{y'}}\right]\ket{\psi_x}
\bra{\psi_{y}}\Lambda(s) \left[\ketbra{\psi_{x_0}}{\psi_{x_0}}\right] \ket{\psi_{y'}},
\nonumber
\end{eqnarray}
so that the `diagonal terms' (with $y=y'$) cancel out with the second contribution and the violation of the homogeneous Chapman-Kolmogorov condition is given by
\begin{eqnarray}
&& Q^{\hat{X}}_{1|1}\left\{x,t|x_0,0\right\} - \sum_{y}   Q^{\hat{X}}_{1|1}(x, t-s | y, 0)  Q^{\hat{X}}_{1|1}(y, s | x_0, 0)  \nonumber\\
&&= \sum_{y\neq y'}\bra{\psi_x} \Lambda(t-s)\left[\ket{\psi_{y}}\bra{\psi_{y'}}\right]\ket{\psi_x}
\bra{\psi_{y}}\Lambda(s) \left[\ketbra{\psi_{x_0}}{\psi_{x_0}}\right] \ket{\psi_{y'}},
\end{eqnarray}
which implies that such difference is equal to 0 for any $x_0, x, t\geq s$ if and only if the Lindblad dynamics is NCGD, see Eq.\eqref{eq:cop2}.
\end{proof}

Now, Theorem 1 easily follows from the following Lemma.
\begin{lemma}
Consider a dichotomic observable $\hat{X}$ with values in $\left\{0,1\right\}$ and
such that the conditional probabilities $Q^{\hat{X}}_{1|1}\left\{x,t ;x_0, 0\right\}$ satisfy Eq.(6) of the main text,
then the correlation function $C_X(t,0)$ satisfies the LGtI
\begin{equation}\label{eq:lgti}
\left |2C_X(t,0) - C_X(2t,0)\right |  \leq \langle X(0) \rangle.
\end{equation}
\end{lemma}
\begin{proof}

First note that
$$
C_X(t,0) := \sum_{x, x_0 = 0,1} Q^{\hat{X}}_2\left\{ x, t;  x_0, 0\right\} x * x_0= Q^{\hat{X}}_2\left\{ 1, t;  1, 0\right\};
$$
since the dichotomic observable has values in $\left\{0,1\right\}$;
thus
\begin{eqnarray}
&& \left |2C_X(t,0) - C_X(2t,0)\right | =\left |2 Q^{\hat{X}}_2\left\{ 1, t;  1, 0\right\} - Q^{\hat{X}}_2\left\{ 1, 2t;  1, 0\right\} \right|  \nonumber\\
 &&= \left | 2Q^{\hat{X}}_{1|1}\left\{1,t |1, 0\right\}-Q^{\hat{X}}_{1|1}\left\{1,2t| 1,0\right\}\right| Q^{\hat{X}}_1\left\{1,0\right\}\nonumber\\
 &&=\left|2Q^{\hat{X}}_{1|1}\left\{1,t |1, 0\right\}-\sum_{x_kk=0,1}  Q^{\hat{X}}_{1|1}\left\{1,t |x_k, 0\right\} Q^{\hat{X}}_{1|1}\left\{x_k,t |1, 0\right\} \right| Q^{\hat{X}}_1\left\{1,0\right\},
\end{eqnarray}
which provides us with Eq.(\ref{eq:lgti}), since $\langle X(0) \rangle = Q^{\hat{X}}_1\left\{1,0\right\}$, while
$2Q^{\hat{X}}_{1|1}\left\{1,t |1, 0\right\}-\sum_{x_kk=0,1}  Q^{\hat{X}}_{1|1}\left\{1,t |x_k, 0\right\} Q^{\hat{X}}_{1|1}\left\{x_k,t |1, 0\right\}$ is maximized by 1 for 
$Q^{\hat{X}}_{1|1}\left\{1,t |1, 0\right\} =1$ or $Q^{\hat{X}}_{1|1}\left\{1,t |1, 0\right\} = 0$ and $Q^{\hat{X}}_{1|1}\left\{0,t |1, 0\right\} = 1$ (as
seen using $Q^{\hat{X}}_{1|1}\left\{1,t |0, 0\right\} = 1-Q^{\hat{X}}_{1|1}\left\{1,t |1, 0\right\}$).
\end{proof}

Note that Lemma 2 holds independently of whether the conditional probabilities
are referring to the quantum setting (and hence are defined as in Eq.(1) of the main text)
or are directly referring to a classical theory: our proof
goes along the same line as that in \cite{Lambert2010},
simply adapting it to (possibly) quantum conditional probabilities.

\section{Proof of Theorem 2 and Theorem 3.}\label{app:semi2}

Before presenting the proof to Theorem 2, let us give the basic idea behind it.
The time-homogeneous Chapman-Kolmogorov equation holds for any classical time-homogeneous Markov
process, that is, Markovianity and classicality imply Chapman-Kolmogorov; note
that the time-homogeneity of the statistics holds, as a consequence of (2)M and the Lindblad
dynamics [see Eq.\eqref{eq:qhat}].
For the converse, we can exploit the definition of jM, which provides us with a notion of Markovianity beyond classical processes,
i.e., for any quantum statistics.
As said, the Markov property (both for classical
and nonclassical statistics) connects the multi-time probability distributions to the initial one-time distribution
and the conditional probability $Q^{\hat{X}}_{1|1}$; as a direct consequence of this, it is then easy to see that, if the Chapman-Komogorov equation holds,
jM directly turns into jCL. Theorem 2 thus follows from the equivalence established in Proposition 1. 

Explicitly, both the Theorems 2 and 3 directly follow from the following Lemma.

\begin{lemma} 
Given a non-degenerate observable 
$\hat{X} = \sum_{x} x \ketbra{\psi_x}{\psi_x}$ and a jM hierarchy of probabilities,
the latter defines a jCL statistics for any initial diagonal state $\rho(0) = \sum_{x_0} p_{x_0} \ketbra{\psi_{x_0}}{\psi_{x_0}}$ 
if and only if Eq.(6) of the main text holds for the quantum conditional probability $Q^{\hat{X}}_{1|1}\left\{x,t|x_0,0\right\}$.
\end{lemma}

\begin{proof}
``Only if": the statistics is, in particular, 2CL, so that we have, for any $x, t\geq s\in \mathbbm{R}^+$,
$$
\sum_{y} Q^{\hat{X}}_{2}\left\{x,t;y,s\right\} =  Q^{\hat{X}}_{1}\left\{x,t\right\}.
$$
But then, using the definition of conditional probability $ Q^{\hat{X}}_{2}\left\{x,t;y,s\right\} =  Q^{\hat{X}}_{1|1}\left\{x,t|y,s\right\}*
 Q^{\hat{X}}_{1}\left\{y,s\right\}$ and, crucially, the time-homogeneity guaranteed by the 2M and the Lindblad equation (see Eq.\eqref{eq:qhat}), we can write
 $$
 \sum_{y}  Q^{\hat{X}}_{1|1}\left\{x,t-s|y,0\right\}*Q^{\hat{X}}_{1}\left\{y,s\right\}=  Q^{\hat{X}}_{1}\left\{x,t\right\}.
 $$
Using the Kolmogorov condition, this time w.r.t. the initial value,
and the definition of conditional probability, 
the previous relation gives
 $$
 \sum_{x_0}Q^{\hat{X}}_{1}\left\{x_0,0\right\} \left( \sum_{y} Q^{\hat{X}}_{1|1}\left\{x,t-s|y,0\right\}*Q^{\hat{X}}_{1|1}\left\{y,s|x_0,0\right\}- Q^{\hat{X}}_{1|1}\left\{x,t|x_0,0\right\}\right) = 0,
 $$
which directly provides us with the Chapman-Kolmogorov composition law in Eq.(6) of the main text, since, by assumption, 
we can choose any initial diagonal state $\rho(0) = \sum_{x_0} p_{x_0} \ketbra{\psi_{x_0}}{\psi_{x_0}}$,
and then any distribution of $Q^{\hat{X}}_{1}\left\{x_0,0\right\} = p_{x_0}$.\\ 

``If": Eq.(6) of the main text for the quantum conditional probability $Q^{\hat{X}}_{1|1}\left\{x,t|x_0,0\right\}$ means that
\begin{eqnarray}
  \mbox{tr}\left\{\mathcal{P}_x e^{\mathcal{L} t} \left[\ketbra{\psi_{x_0}}{\psi_{x_0}}\right]\right\} = 
 \sum_{x_k}  \mbox{tr}\left\{\mathcal{P}_x e^{\mathcal{L} (t-s)} \left[\ketbra{\psi_{x_k}}{\psi_{x_k}}\right]\right\}
 \mbox{tr}\left\{\mathcal{P}_{x_k} e^{\mathcal{L} s} \left[\ketbra{\psi_{x_0}}{\psi_{x_0}}\right]\right\} \nonumber\\ 
 \forall x_0, x; t\geq s \in \mathbbm{R}^+,\nonumber
\end{eqnarray}
which, replaced in Eq.\eqref{eq:eq}, implies Eq.(2) of the main text [$s \mapsto t_k-t_{k-1}, x_0\mapsto x_{k-1}, t \mapsto t_{k+1}-t_{k-1}, x \mapsto x_{k+1}$],
so that if the hierarchy is jM it will also be jCL; note that this is the case also for $k=1$ since we assume the initial state to be diagonal in the selected basis.
\end{proof}

Theorem 2 hence directly follows from Lemma 3 and Proposition 1, while Theorem 3 follows, e.g., from Lemmas 2 and 3.

\section{Difference between the various types of incoherent operations}\label{app:ch}

Let us start by presenting a very simple example of an evolution which is CGD, i.e., which can generate and detect coherence
[see Eq.(5) of the main text] and which, in addition, allows us to connect
such property with nonclassicality in an easy way.
Hence, consider a two-level system, initially in the eigenstate $\ket{-1}$ of $\hat{\sigma}_z$
(where $\left\{\ket{\ell}\right\}_{\ell=-1,1}$ is the eigenbasis of the operator $\hat{\sigma}_z$)
and evolving under the unitary operators
\begin{equation}\label{novel}
U(t) = e^{-i \hat{\sigma}_y t}.
\end{equation}
Indeed, this evolution can generate coherence with respect to $\hat{\sigma}_z$ at an intermediate time, i.e., there are times $t_1>0$
when $|\bra{\pm1}U(t_1)\ket{-1}\bra{-1}U^{\dag}(t_1)\ket{\mp1}| \neq 0$, and can convert such coherence
into populations, i.e. there are times $t_2$ such that $\bra{1}U(t_2-t_1)\ket{\mp1}\bra{\pm1}U^{\dag}(t_2-t_1)\ket{1} \neq 0$,
which implies that the condition in Eq.(5) of the main text is satisfied: the evolution is CGD; see
also Fig.\ref{fig:examplencgdcoherent} {\bf a)}.
An important point of this example is that, since the evolution is unitary, the QRT is automatically guaranteed
[i.e, Eq.(1) and Eq.(4) of the main text coincide], so that
the two-time statistics associated with measurements of the observable $\hat{\sigma}_z$
is simply given by
$$
Q_2^{\hat{\sigma}_z}\left\{1, t_2; x, t_1\right\} = \left|\bra{1}U(t_2-t_1)\ket{x}\right|^2
\left|\bra{x}U(t_1)\ket{-1}\right|^2,
$$
with $x=-1,1$. It is then easy to see that the CGD condition implies
$\sum_x Q_2^{\hat{\sigma}_z}\left\{1, t_2; x, t_1\right\} \neq Q_1^{\hat{\sigma}_z}\left\{1, t_2\right\} =\left|\bra{1}U(t_2)\ket{-1}\right|^2$,
so that the multitime statistics is nonclassical (see Definition 1 in the main text):
the coherence generated at intermediate time $t_1$ and turned into population at time $t_2$
directly implies a violation of the Kolmogorov condition, i.e., nonclassicality.
This is not surprising, since, by virtue of the Theorem 2, we know that 
nonclassicality and CGD of the system's (nondegenerate) observables coincide,
whenever the QRT holds.
If this is not the case, the correspondence between nonclassicality and CGD 
is no longer guaranteed, as we will exemplify in the next Section.

As mentioned in the main text, it is trivial to see that the maximally incoherent operations (MIO)~\cite{aaberg2014catalytic} and the coherence nonactivating maps~\cite{Liu2017} are subsets of the NCGD maps. What we show here is that they are strict subsets, by giving an explicit example. Consider the
completely positive and trace preserving map
acting on a basis of linear operators on $\mathbbm{C}^2$ as
\begin{equation}\label{ex:mapncgdnotinc}
 \Lambda: 
\begin{array}{cc}
\left(
\begin{array}{cc}
1 & 0 \\
0& 0 \\
\end{array}
\right) & \left(
\begin{array}{cc}
0 & 1 \\
0 & 0 \\
\end{array}
\right) \\
\left(
\begin{array}{cc}
0 & 0 \\
1 & 0 \\
\end{array}
\right) & \left(
\begin{array}{cc}
0 & 0 \\
0 & 1 \\
\end{array}
\right) \\
\end{array}
\mapsto
\begin{array}{cc}
\left(
\begin{array}{cc}
0.996 & -0.003 i \\
0.003 i & 0.004 \\
\end{array}
\right) & \left(
\begin{array}{cc}
0.003 & 0.99 \\
0 & -0.003 \\
\end{array}
\right) \\
\left(
\begin{array}{cc}
0.003 & 0 \\
0.99 & -0.003 \\
\end{array}
\right) & \left(
\begin{array}{cc}
0.004 & 0.003 i \\
-0.003 i & 0.996 \\
\end{array}
\right) \\
\end{array}.
\end{equation}
The map is NCGD, while it  both creates coherence and also is able to detect it. Explicitly:
\begin{equation}\label{eq:NCGD}
	\left\|
		\Delta \circ \Lambda \circ \Lambda \circ \Delta -\Delta \circ\Lambda \circ \Delta \circ \Lambda\circ \Delta
		\right\|_{\infty}=0,
\end{equation}
\begin{equation}\label{eq:INC}
\left\|
\Lambda \circ \Delta -\Delta \circ \Lambda\circ \Delta
\right\|_{\infty}=0.003,
\end{equation}
\begin{equation}\label{eq:CNA}
\left\|
\Delta \circ \Lambda-\Delta \circ \Lambda\circ \Delta
\right\|_{\infty}=0.003,
\end{equation}
where $\| \Lambda\|_\infty$denotes the
infinity norm of the $4 \times 4$ matrix
given by the action of $\Lambda$
on the basis of operators on $\mathbbm{C}^2$;
recall that the infinity norm is the
maximum among the absolute sums of the columns. 
Indeed, the same is true for applying the map multiple times: 
the NCGD condition remains fulfilled, while the above norm increases to over 0.12, as shown in Fig.~\ref{fig:examplencgdcoherent} {\bf c)}.

\begin{figure}[h]
	\centering
	{\bf (a)}\hskip5.9cm {\bf (b)}\hskip5.9cm{\bf (c)}\\
	\vspace{0.2cm}
	\includegraphics[width=0.25\linewidth]{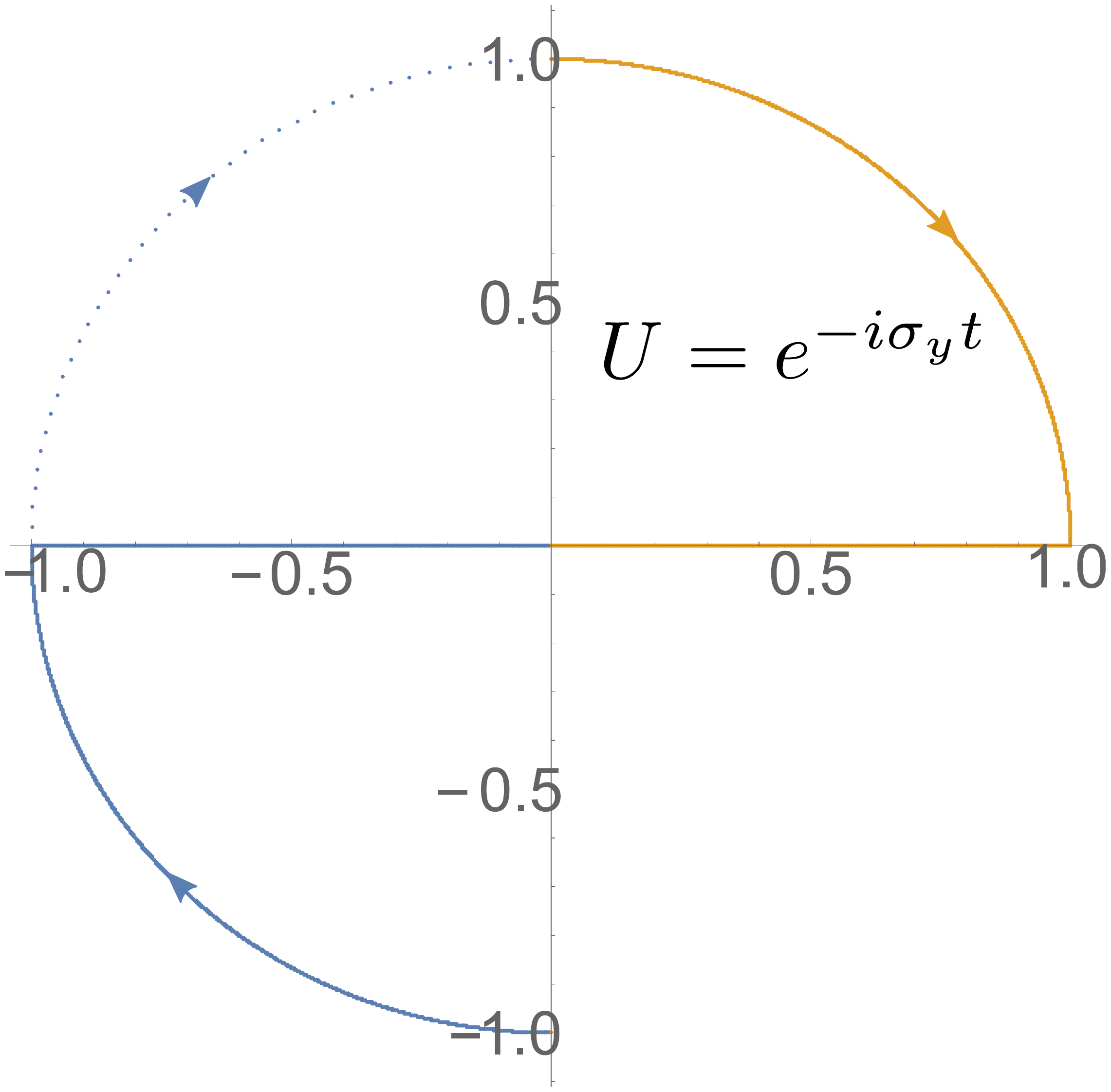}\hspace{0.05\linewidth}
	\includegraphics[width=0.27\linewidth]{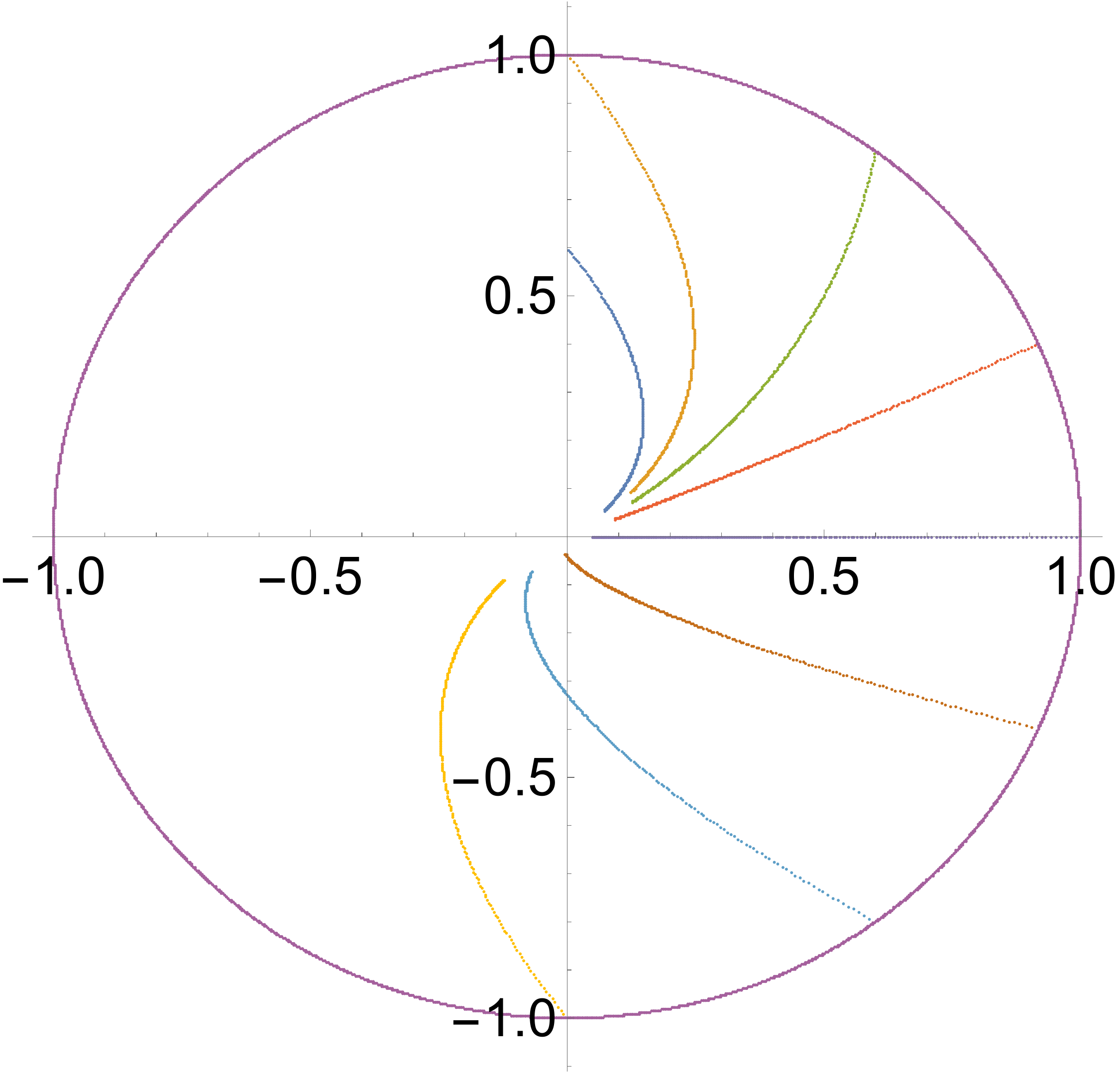}\hspace{0.05\linewidth}\includegraphics[width=0.33\linewidth]{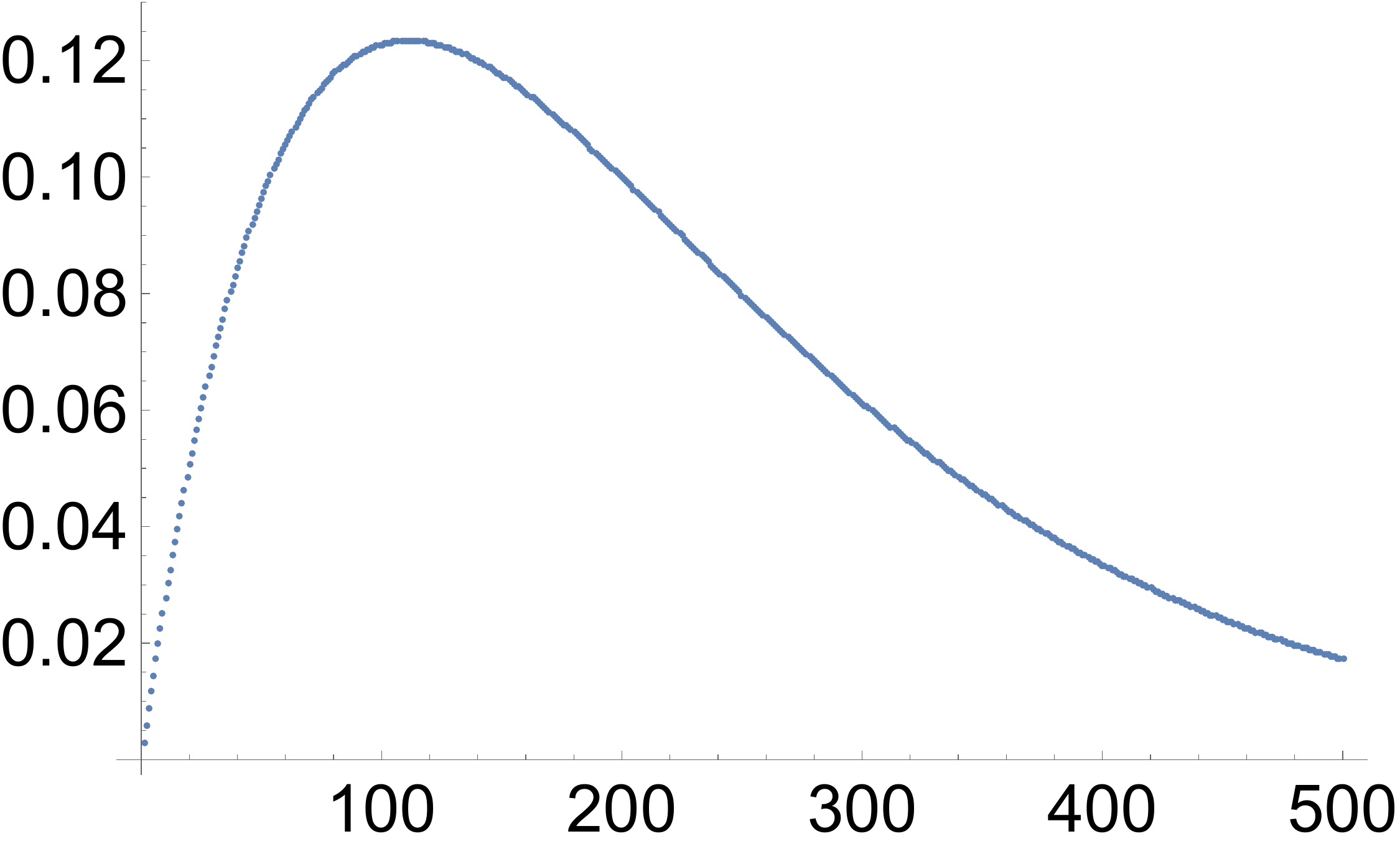}
	\caption{{\bf a)} Example of an evolution fixed by Eq.\eqref{novel} and by a possible measurement of $\hat{\sigma}_z$
	at the intermediate time $t_1=\pi/4$: starting from $\ket{-1}$, the system's state
	follows the solid blue line until $t_1$. If no measurement is performed,
	the system will then follow the dotted blue line, completely turning the coherence with respect to $\hat{\sigma}_z$
	generated at $t_1$ into the population $\ket{1}$ at $t_2=\pi/2$. On the other hand,
	if a measurement at $t_1$ is performed, the state will be projected to $\ket{1}$ or $\ket{-1}$,
	with probability $1/2$ each; after that, the state will continue
	its evolution along the solid orange or blue line, respectively, so that a second measurement at time $t_2=\pi/2$
	would still project the system's state to $\ket{1}$ or $\ket{-1}$,
	with equal probabilities $1/2$. But then, we clearly have 
	$1/2=\sum_{x=\pm1} Q_2^{\hat{\sigma}_z}\left\{1, \pi/2; x,\pi/4 \right\} \neq Q_1^{\hat{\sigma}_z}\left\{1, \pi/2\right\}=1$,
	i.e., the multitime statistics is nonclassical.
	 {\bf b)} Evolution of seven states on the XZ plane of the Bloch sphere under 1 to 300 applications of a map (see Eq.~\eqref{ex:mapncgdnotinc}) that is NCGD, but neither coherent nor coherence nonactivating. {\bf c)} 
	Norm of the difference between the coherent map in 
	Eq.~\eqref{ex:mapncgdnotinc} and the incoherent one (see Eq.~\eqref{eq:INC}) as a function of the number of applications. In this example this is the same as the norm of the difference between the coherence nonactivating and the actual evolution (see Eq.~\eqref{eq:CNA}).}
	\label{fig:examplencgdcoherent}
\end{figure}

\section{More details on the examples of NM multi-time statistics}\label{app:ex}
In this Section, we provide the explicit calculations for the examples discussed in the main text.
Let us emphasize that the model taken into account is very similar to the photonic system considered in \cite{Liu2013}
to realize experimentally the transition between a Markovian and a non-Markovian dynamics, and further exploited in \cite{Guarnieri2014}
to investigate the link between the dynamical notions of Markovianity and the QRT for system's observables (rather than for the projective maps resulting
from ideal selective measurements, as done here). There, the open system is given by the polarization degree of freedom
of a photon, while the continuous degree of freedom represents the frequency of the photon, $\omega \in \mathbbm{R}^+$. Here, instead, we
are taking into account a continuous degree of freedom, say the particle's momentum, which can take on negative values, $p \in \mathbbm{R}$,
for a reason which will become clear in a moment.
Furthermore, essentially the same model (with coupling to position)
has been studied in the context of dynamical decoupling in \cite{Arenz2015}.

Before proceeding, let us introduce more general definitions, which can be applied also to the non-Lindblad
dynamics encountered in the following examples (in particular,
in the second one, referring to an initial momentum distribution given by the superposition of two Gaussians).
To do so, we need the notion of propagators of the dynamics, that is, the maps $\Lambda(t,s)$ (not necessarily completely positive, neither positive)
such that 
\begin{equation}\label{eq:prop}
\Lambda(t) = \Lambda(t,s) \Lambda(s),
\end{equation}
for any $t \geq s \in \mathbbm{R}^+$,
so that $\rho(t) = \Lambda(t,s)[\rho(s)]$.
For the sake of simplicity, we assume that the dynamics is divisible \cite{Breuer2016}, i.e., that the propagators can be defined as in Eq.\eqref{eq:prop}
for any $t\geq s\in \mathbbm{R}^+$; indeed, this is always the case for a Lindbladian dynamics,
where the propagators are defined by $\Lambda(t,s)=\Lambda(t-s,0)=e^{\mathcal{L}(t-s)}$. 

We say that the multi-time statistics is j-Markovian if it satisfies the QRT theorem with respect to the dynamical maps and the corresponding
propagators, i.e., if 
\begin{eqnarray}\label{eq:LQRTn}
&& Q^{\hat{X}_S}_n\left\{x_n, t_n \ldots x_1,t_1\right\}   = \mbox{tr}_S\left\{\mathcal{P}_{x_n}\Lambda(t_n,t_{n-1})\ldots \mathcal{P}_{x_1}\Lambda(t_1)\rho_S(0)\right\},
\end{eqnarray}
for any $n\leq j, x_1, \ldots x_n, t_n \geq \ldots t_1$.
Of course, this definition reduces to that in Eq.(4) of the main text if we consider a Lindblad dynamics.
The validity of the QRT and a Lindblad dynamics ensure that the statistics is 
time-homogeneous, i.e., that
the two-time conditional probabilities satisfy 
$Q^{\hat{X}}_{1|1}\left\{x,t|y,s\right\} = Q^{\hat{X}}_{1|1}\left\{x,t-s|y,0\right\}$. 
Hence, we can see the definition in Eq.\eqref{eq:LQRTn} as the general definition of Markovian, possibly time-inhomogeneous
statistics, while the definition in Eq.(4) of the main text corresponds to the Markovian time-homogeneous case.

Finally, we say that
a (divisible) dynamics $\left\{\Lambda(t)\right\}_{t\in \mathbbm{R}^+}$,
with propagators $\Lambda(t,s)$ is CGD whenever
there exist instants $t\geq s \geq r\in \mathbbm{R}^+$ such that
	\begin{equation}\label{eq:dellam2}
	\Delta \circ \Lambda(t,s)\circ \Delta \circ \Lambda(s,r)\circ\Delta \neq \Delta \circ \Lambda(t,r)\circ\Delta;
	 \end{equation}
once again, this reduces to the definition given in the main text [Eq.(5)] for a Lindblad dynamics.	

\subsection{General expressions}
Given the unitary
\begin{equation}
U(t) \ket{\ell, p} = e^{i \ell p t} \ket{\ell, p}
\end{equation}
(where, as in the main text, $\left\{\ket{\ell}\right\}_{\ell=-1,1}$ is the eigenbasis of the system operator $\hat{\sigma}_z$),
one can straightforwardly evaluate the exact expression of the multi-time joint probability distribution, 
as given by Eq.(1) of the main text, as well as that provided by the QRT, see Eq.\eqref{eq:LQRTn}, i.e.,
the one which is given by a Markovian description of the multi-time statistics.
For the sake of simplicity, we focus on the two-time statistics and we take as initial condition $\rho(0) = \ket{+}\bra{+}\otimes \ket{\varphi_E}\bra{\varphi_E}$, with
$
\ket{\varphi_{E}} = \int_{-\infty}^{\infty} d p f(p) \ket{p}$ ($\int_{-\infty}^{\infty} d p |f(p)|^2 =1$). It is useful to define $k(t) = \int_{-\infty}^{\infty} d p |f(p)|^2 e^{2 i p t}$.
Now, let us also fix that both the first and the second outcomes of the measurement of $\hat{\sigma}_x$
yield the same result, $+$, so that we have
\begin{eqnarray}\label{eq:nonqrt}
 Q_2^{\hat{\sigma}_x}\{+,t;+,s\}&=&\mbox{tr}_E\{\bra{+}{\mathcal{U}}(t-s)[\ket{+}\bra{+}{\mathcal{U}}(s)[\ket{+}\bra{+}\otimes \ket{\varphi_E}\bra{\varphi_E}]\ket{+}\bra{+}]\ket{+} \} \nonumber\\
 &=& \frac{1}{4}\Re{\left[\frac{1}{2}k(t-2s)+\frac{1}{2}k(t)+k(t-s)+k(s)+1\right]},
\end{eqnarray}
where $\Re$ denotes the real part.
Moreover, since the map and the propagator of the above dynamics are given by
	\begin{equation}
	\Lambda(t)[\rho]=\left(
	\begin{array}{cc}
	\rho_{-1 -1} &  k(t)\rho_{-11}\\ 
	k^*(t)\rho_{1-1} & \rho_{11}
\end{array}
	\right)
	\end{equation}
	and
	\begin{equation}
	\Lambda(t,s)[\rho]=\Lambda(t)\circ\Lambda^{-1}(s)[\rho]=\left(
	\begin{array}{cc}
	\rho_{-1-1} &  \frac{k(t)}{k(s)}\rho_{-11}\\ 
	\frac{k^*(t)}{k^*(s)}\rho_{1-1} & \rho_{11}
	\end{array}
	\right),
	\end{equation}
we get, see Eq.\eqref{eq:LQRTn},
\begin{eqnarray}
Q_{2M}^{\hat{\sigma}_x}\{+,t;+,s\}&=&\bra{+}\Lambda(t,s)\Big[\ket{+}\bra{+}\Lambda(s)[\ket{+}\bra{+}]\ket{+}\bra{+}\Big]\ket{+} \nonumber\\
&=&\frac{1}{4}\Re\left[\frac{1}{2}k(t)\left(\frac{k^*(s)}{k(s)}+1\right)+\frac{k(t)}{k(s)}+k(s)+1\right].
 \end{eqnarray}
This means that the Markovian description implies $Q_2^{\hat{\sigma}_x}\{+,t;+,s\} = Q_{2M}^{\hat{\sigma}_x}\{+,t;+,s\}$, i.e.,
\begin{eqnarray}
\Re\left[\frac{1}{2}k(t)\frac{k^*(s)}{k(s)}+\frac{k(t)}{k(s)}\right] = \Re\left[\frac{1}{2}k(t-2s)+k(t-s)\right];
\end{eqnarray}
indeed, a violation of this condition will be enough to prove that the statistics is NM.

In order to check whether the Kolmogorov condition in Eq.(2) of the main text holds for the two-time probabilities,
we need also to evaluate the other two-time probability distribution
\begin{eqnarray}
 Q_2^{\hat{\sigma}_x}\{+,t;-,s\}&=& \mbox{tr}_E\{\bra{+}{\mathcal{U}}(t-s)\Big[\ket{-}\bra{-}{\mathcal{U}}(s)[\ket{+}\bra{+}\otimes \ket{\varphi_E}\bra{\varphi_E}]\ket{-}\bra{-}\Big]\ket{+} \} \nonumber\\
 &=&\frac{1}{4}\Re[\frac{1}{2}k(t-2s)+\frac{1}{2}k(t)-k(t-s)-k(s)+1],
\end{eqnarray}
as well as the one-time probability
\begin{eqnarray}
Q_1^{\hat{\sigma}_x}\{+,t\}&=&\mbox{tr}_E\{\bra{+}{\mathcal{U}}(t)[\ket{+}\bra{+}\otimes \ket{\varphi_E}\bra{\varphi_E}]\ket{+} \}\nonumber \\
&=&\Re[\frac{1}{2}+\frac{1}{2}k(t)].
\end{eqnarray}
Hence, setting $K_\pm(t,s)\equiv Q_2^{\hat{\sigma}_x}\{\pm,t;+,s\}+Q_2^{\hat{\sigma}_x}\{\pm,t;-,s\}-Q_1^{\hat{\sigma}_x}\{\pm,t\}$,
the statistics is 2-CL if and only if 
\begin{eqnarray}
K_+(t,s) = \frac{1}{4}\Re[k(t-2s)-k(t)]=0;
\end{eqnarray}
note that, since $ Q_2^{\hat{\sigma}_x}\{-,t;\pm,s\} = Q_1^{\hat{\sigma}_x}\{\pm,s\}- Q_2^{\hat{\sigma}_x}\{+,t;\pm,s\}$, one has $K_-(t,s)=- K_+(t,s)$.

Finally, since we are interested in the connection between classicality and coherences, 
we want to check whether the dynamics is (N)CGD.
With $\Delta$ defined with respect to the eigenbasis of $\hat{\sigma}_x$, we have
 \begin{eqnarray}
 \|&&(\Delta\circ\Lambda(t,s)-\Delta\circ\Lambda(t,s)\circ\Delta)\Lambda(s)\Delta\|_\infty
 = \frac{ 2\left|\Im\left[k(s)\right] \Im\left[k^*(s) k(t)\right]\right|}{4 |k(s)|^2} \label{eq:realk}
  \end{eqnarray}
where $\Im$ denotes the imaginary part. 
For the sake of simplicity, 
we set $r=0$ [compare with the general definition in Eq.\eqref{eq:dellam2}], which is in any case enough to detect CGD.

\subsection{Lorentzian distribution}
As said in the main text, for a Lorentzian distribution
\begin{equation}
|f(p)|^2 = \frac{\Gamma}{\pi(\Gamma^2+(p-p_0)^2)},
\end{equation}
the decoherence function simply reduces to 
\begin{equation}
k(t)= e^{2 i p_0 t} e^{- 2 \Gamma t}.
\end{equation}
The latter implies a Lindblad pure dephasing dynamics \cite{Guarnieri2014},
\begin{equation}\label{eq:pudeme}
\frac{d}{dt}\rho(t) = - i p_0 \left[\hat{\sigma}_z, \rho(t)\right] + \Gamma \left(\hat{\sigma}_z \rho(t) \hat{\sigma}_z-\rho(t)\right),
\end{equation}
from which we can read the physical meaning of the two parameters, $p_0$ and $\Gamma$ defining the Lorentzian distribution in this context.

In particular, for $p_0 = 0$, we get
\begin{eqnarray}
Q^{\hat{\sigma}_x}_{2}\left\{+, t; +, s\right\} &=& \frac{1}{4}\left(1+e^{-2 \Gamma s}+e^{-2\Gamma(t-s)}+\frac{1}{2}e^{-2 \Gamma t}\right) \nonumber\\
&&+\frac{1}{8}\left( \cosh\left(2\Gamma(t-2s)\right) - \sinh\left(2 \Gamma |t-2s|\right)\right), \label{eq:qtre}
\end{eqnarray}
while the QRT gives us
\begin{eqnarray}
Q^{\hat{\sigma}_x}_{2M}\left\{+, t; +, s\right\} 
&=& \left\{  \bra{+}e^{\mathcal{L}(t-s)}\left[\ketbra{+}{+}e^{\mathcal{L} s}\left[\ketbra{+}{+}\right]\ketbra{+}{+}\right]\ket{+}\right\} \nonumber\\
&=& \frac{1}{4}\left(1+e^{-2\Gamma s}\right)\left(1+e^{-2\Gamma (t-s)}\right). \label{eq:qqrt}
\end{eqnarray}
As shown in Fig.1{\bf a)} of the main text, these two functions are clearly different, implying that the present statistics is NM,
since the QRT is not satisfied \cite{footboh}.
In addition, the statistics is not even classical, as follows from
\begin{equation}\label{eq:ineq}
\sum_y Q^{\hat{\sigma}_x}_{2}\left\{1, t; y, s\right\} \neq Q^{\hat{\sigma}_x}_{1}\left\{1, t\right\},
\end{equation}
which can be easily shown since one has
$$
\partial_s \sum_y Q^{\hat{\sigma}_x}_{2}\left\{+, t; y, s\right\} = \Gamma \mbox{sgn}\left\{t-2s\right\}e^{-2 \Gamma |t-2s|},
$$
which is of course different from 0, thus guarantying the inequality in Eq.\eqref{eq:ineq}.
For $p_0=0$ the model is furthermore NCGD: Eq.(5) of the main text does not hold for any choice of times and states
spanned by the reference basis. As we have here pure dephasing in the $z$-direction, coherences in the $x$-direction cannot
be even generated.
This example clearly shows how the non-classicality of a NM statistics
might be fully unrelated even from the presence itself of quantum coherence in the dynamics.

The main reason behind that is, as said, the irreducible complexity of the hierarchy of joint probability distributions,
so that two-time probabilities cannot be generally inferred from one-time probabilities, even if the latter follow a homogeneous (Lindblad) \emph{dynamics}.
Note that for the same model, if we consider as initial state of the system the totally mixed state,
$\rho(0)= (\ketbra{+}{+}+\ketbra{-}{-})/2$, one has $Q^{\hat{\sigma}_x}_{2}\left\{x, t; y, s\right\}=Q^{\hat{\sigma}_x}_{2}\left\{x, t-s; y, 0\right\}$,
which then satisfies the QRT, so that the statistics is 2M. 
Nevertheless, in this case the three-time probability distribution $Q^{\hat{\sigma}_x}_{3}$
would not satisfy the QRT  \cite{Lindblad1980,Accardi1982} (essentially, the state after the first selective measurement plays the role which was played by the initial state above):
one has here a statistics which is 2M, but not 3M.
We also note that, for a generic initial environmental 
distribution, such choice of $\rho(0)$ would lead to a 2CL, but not 3CL statistics.
We conclude that one cannot generally go from one-time probability distributions to two-time ones;
and even if one can go from one- to two-time probabilities, the three-time probability distributions might not be deducible from the lower ones.
Finally, we leave as an open question whether it is possible to find more elaborated examples where the statistics
can be $j$CL ($j$M), but not $j+1$CL ($j+1$M) for a generic $j$. In any
case, the whole hierarchies introduced in Definitions 1 and 2 of the main text are useful, both for
the practical reasons mentioned in the main text (one might be able to access $j$-time, but
not $j+1$-time statistics) and because one might want to quantify different degrees of 
non-classicality (non-Markovianity) in the multi-time statistics: even if 
the statistics is neither $j$CL nor $j+1$CL ($j$M nor $j+1$M), 
it can be of interest to speak of a stronger violation of, say, $j$CL rather than $j+1$CL
($j$M rather than $j+1$M), e.g., by comparing the different deviations from the corresponding Kolmogorov conditions in Eq.(2)
of the main text (the QRT in Eq.(4) of the main text).

\subsection{Superposition of Gaussians}
For a distribution given by the sum of two Gaussians
\begin{equation}
|f(p)|^2= \sum_{i=1,2} A_i e^{-\frac{(p-p_i)^2}{2 \sigma^2_i}}
\end{equation}
where $A_1=\frac{1}{\sqrt{2\pi}\sigma(1+A_\theta)}$, $A_2=A_\theta A_1$ and $\sigma_1=\sigma_2=\sigma$, the decoherence function reduces to
\begin{equation}
k(t)=\frac{e^{-2\sigma^2t^2}}{A_\theta+1}\left[e^{2 i p_1t}+A_\theta e^{2 i p_2t}\right].
\end{equation}
For the specific choice of parameters $A_\theta=\sigma=p_1=t=1$, $p_2=2p_1$, the functions $Q_2^{\hat{\sigma}_x}\{+,t;+,s\}$ and $Q^{\hat{\sigma}_x}_{2M}\left\{+, t; +, s\right\}$ are, in general, different. The present statistics is thus NM, as shown in Fig.\ref{fig:cgdcl} \textbf{a)}.
\begin{figure}[h]
	\centering
		\hspace{-7cm}{\bf\large{(a)}}\\
	\includegraphics[width=0.4\linewidth]{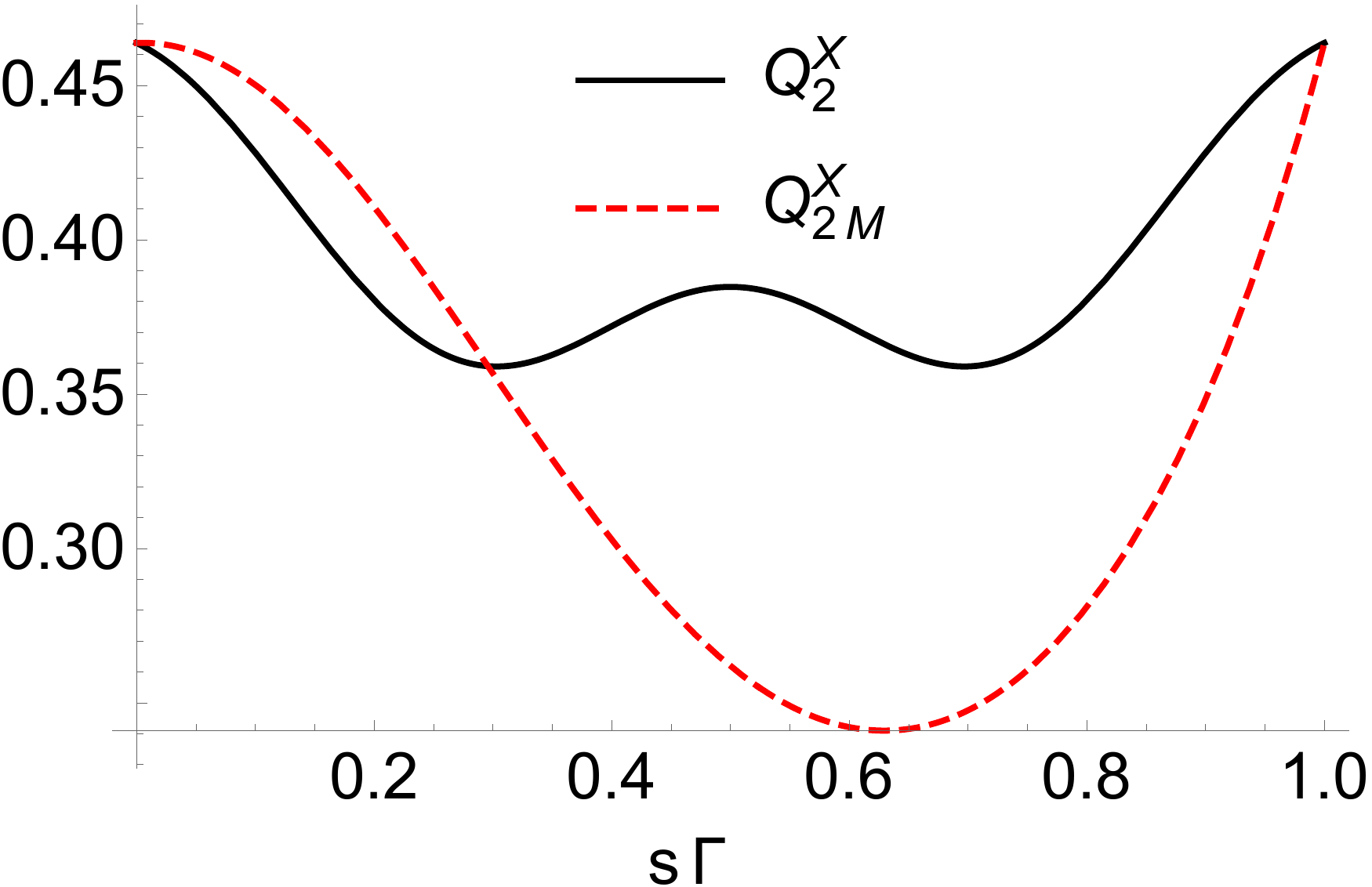}\hspace{0.1\linewidth}\includegraphics[width=0.4\linewidth]{Fig2b}
	\caption{\textbf{a)} Comparison of the two-time probability distribution as a function of the first time s, for
	the exact formula (solid, black line) and the 2M statistics (dashed, red line). \textbf{b)} Violation of 2CL ($K_+(t,s)$) (dashed, blue line) and violation of NCGD ($N(t,s)$) (solid, red line). For both plots the environmental distribution is given by a sum of two Gaussians, with (in arbitrary units) $A_\theta=\sigma=p_1=t=1$, $p_2=2p_1$, $s\in[0,t]$.}
	\label{fig:cgdcl}
\end{figure}

In order for the statistics to be 2-CL, the following condition must hold
\begin{eqnarray}
K_+(t,s) &=& \frac{1}{4}\Re[k(t-2s)-k(t)]\nonumber\\
 &=& \frac{e^{-2\sigma^2(t-2s)^2}}{4(A_\theta+1)}\left[\cos(2p_1(t-2s))+A_\theta\cos(2p_2(t-2s))\right]\nonumber
 \\&&+\frac{e^{-2\sigma^2t^2}}{4(A_\theta+1)}\left[\cos(2p_1t)+A_\theta\cos(2p_2t)\right]=0.
\end{eqnarray}

Furthermore, the model is CGD if [see Eq.\eqref{eq:realk}]
the quantity
\begin{eqnarray}
  &&N(t,s):= \frac{ 2\left|\Im\left[k(s)\right] \Im\left[k^*(s) k(t)\right]\right|}{4 |k(s)|^2}  \\
  &&=\frac{1}{2 (1 + 
   A_\theta) \left|
  e^{2 i s p_1} + A_\theta e^{2 i s p_2}\right|^2}
  \left(e^{-2 t^2 \sigma^2} \left[\sin(2 s p_1) + 
   A_\theta \sin(2 s p_2)\right]\right. \nonumber\\
 &&  \left.\phantom{\frac{1}{1}}\times\left[\sin(2 (s - t) p_1) + 
   A_\theta (A_\theta \sin(2 (s - t) p_2) - 
      \sin(2 t p_1 - 2 s p_2) + 
      \sin(2 s p_1 - 2 t p_2))\right]\right)\nonumber
\end{eqnarray}
is different from 0.

As can be seen in Fig.\ref{fig:cgdcl} \textbf{a)}, for the considered choice of parameters the dynamics is NM at instants different from $s=0.29$. This allows for the existence of scenarios where the possible classicality of the statistics is unrelated to the absence of coherences. As a matter of fact, at the specific instants $s=0.21$ and $s=0.79$, where QRT is not satisfied, one finds that $K_+(t,s)=0$ and $N(t,s)\neq 0$, implying that 2-CL holds together with CGD (Fig.\ref{fig:cgdcl} \textbf{b)}).
By investigation (not reported here) of the model at hand in a wide region of parameters, we also observe that there does not seem to be a threshold in the amount of CGD, above which the violation of 2CL is guaranteed.


\begin{thebibliography}{10}
\bibitem{foot5}
From now on, we drop the subfix $S$.
\end{thebibliography}
\end{document}